\documentclass[runningheads,a4paper]{llncs}
\usepackage{amssymb}
\usepackage{hyperref}
\usepackage[toc,page]{appendix}
\hypersetup{
	colorlinks   = true,    % Colours links instead of ugly boxes
	urlcolor     = blue,    % Colour for external hyperlinks
	linkcolor    = blue,    % Colour of internal links
	citecolor    = red      % Colour of citations
}
\usepackage{url}
\usepackage[toc,page]{appendix}
\usepackage{epsfig}
\usepackage{graphicx}
\usepackage{tcolorbox}
\usepackage{tikz}
\usepackage{hyperref} % For hyperlinks in the PDF
\usepackage{amsmath} %%,amsthm}
\usepackage{cleveref}
\usepackage[shortlabels]{enumitem}
\usepackage[linesnumbered,ruled,vlined]{algorithm2e}
\usepackage{cite}
\usetikzlibrary{shadows.blur}
\usetikzlibrary{backgrounds}
\usetikzlibrary{patterns}
\definecolor{my_red}{RGB}{208, 2, 27}
\definecolor{my_green}{RGB}{126, 211, 33}
\definecolor{my_blue}{RGB}{74, 144, 226}
\definecolor{my_purple}{RGB}{144, 19, 254}
\definecolor{my_olive}{RGB}{174, 180, 51}
\makeatletter
\newcommand*{\rom}[1]{\expandafter\@slowromancap\romannumeral #1@}

\makeatother

\begin{document}

\mainmatter  % start of an individual contribution

\title{Finding All Leftmost Separators of Size $\leq k$}

\titlerunning{Finding All Leftmost Separators}

\author{Mahdi Belbasi, Martin F\"urer}
%\thanks{Research supported in part by NSF Grant CCF-1320814.}}
%
\authorrunning{Mahdi Belbasi, Martin F\"urer}
%%\authorrunning{Lecture Notes in Computer Science: Authors' Instructions}

\institute{Department of Computer Science and Engineering \\
	Pennsylvania State University \\
	University Park, PA 16802,  USA \\
	\{belbasi, fhs\}@psu.edu 
}


%%\toctitle{Lecture Notes in Computer Science}
%%\tocauthor{Authors' Instructions}
\maketitle

\begin{abstract}
	We define a notion called leftmost separator of size at most $k$. A leftmost separator of size $k$ is a minimal separator $S$ that separates two given sets of vertices $X$ and $Y$ such that we ``cannot move $S$ more towards $X$'' such that $|S|$ remains smaller than the threshold. One of the incentives is that by using leftmost separators we can improve the time complexity of treewidth approximation. Treewidth approximation is a problem which is known to have a linear time FPT algorithm in terms of input size, and only single exponential in terms of the parameter, treewidth. It is not known whether this result can be improved theoretically. However, the coefficient of the parameter $k$ (the treewidth) in the exponent is large. Hence, our goal is to decrease the coefficient of $k$ in the exponent, in order to achieve a more practical algorithm. Hereby, we trade a linear-time algorithm for an $\mathcal{O}(n \log n)$-time algorithm. The previous known $\mathcal{O}(f(k) n \log n)$-time algorithms have dependences of $2^{24k}k!$, $2^{8.766k}k^2$ (a better analysis shows that it is $2^{7.671k}k^2$), and higher. In this paper, we present an algorithm for treewidth approximation which runs in time $\mathcal{O}(2^{6.755k}\  n \log n)$,
	
	Furthermore, we count the number of leftmost separators and give a tight upper bound for them. We show that the number of leftmost separators of size $\leq k$ is at most $C_{k-1}$ (Catalan number). Then, we present an algorithm which outputs all leftmost separators in time $\mathcal{O}(\frac{4^k}{\sqrt{k}}n)$.
\end{abstract}

\section{Introduction}
Finding vertex separators that partition a graph in a ``balanced'' way is a crucial problem in computer science, both in theory and applications. For instance, in a divide and conquer algorithm, most of the time it is vital to have balanced subproblems. If we want to separate two subsets of vertices in a graph, we prefer the separator to be closer to the bigger side. %If one of these sets is bigger than the other, we look for a separator as close as possible to that subset. 
In this work, we place the bigger set on the left side and the smaller one on the right. Before going into depth, we review and introduce some notations.

\subsection{Notation}
W.l.o.g., assume that $G$ is a connected graph. $S\subseteq V$ is a separator that separates two subsets of vertices $X, Y\subseteq V$ in $G$, if there is no $X-Y$ path in $G[V\setminus S]$, where $G[V \setminus S]$ is the induced graph on $V \setminus S$. In the following, we use $G - S$ instead of $G[V\setminus S]$ for the sake of simplicity. We call $S$ an $(X,Y)^{G}$-separator. Later on, we drop the superscripts if it is obvious from the context. 
\begin{definition}
	$S^G_{X,Y}$ is the set of all $(X,Y)^G$-separators. 
\end{definition}
\begin{definition}
	The separator $S \in S^G_{X,Y}$ partitions $G - S$ into three sets $V_{X, S}$, $V_{S, Y}$, and $V_{Z}$, where $V_{X, S}$ is the set of vertices with a path from $X\setminus S$, $V_{S, Y}$ is the set of vertices with a path from $Y\setminus S$, and $V_{Z}$ is the set of all vertices  reachable from neither $X\setminus S$ nor $Y\setminus S$ in $G - S$.  
\end{definition}
Having a non-empty $V_{Z}$ set is only to our advantage.
We think of $X$ being on the left side and $Y$ on the right side of $S$. Any of the three sets ($X$, $Y$, and $S$) might intersect.
\begin{definition}\textbf{Partial Ordering.}
	We say separator $S \in S^G_{X,Y}$ is at least as much to the left as separator $S' \in S^G_{X,Y}$ if\/  $V_{X, S} \subseteq V_{X, S'}$. In this case, we use the notation $S \preceq S'$.
\end{definition}
\begin{definition}
	Separator $S \in S^G_{X,Y}$ is called an \emph{($X$, $Y$, $\leq k$)$^G$-separator} if $|S| \leq k$.
\end{definition}
\begin{definition}\label{def:leftmost}
	Separator $S \in S^G_{X, Y}$ is called a \emph{leftmost ($X$, $Y$, $\leq k$)$^G$-separator} if it is minimal and there exists no other minimal $(X$, $Y$, $\leq k)^G$-separator $S'$ such that $S' \preceq S$.
	\begin{comment}
	\footnote{Some call it \emph{important separator} like in \cite{chitnis2013fixed},  and \cite{lokshtanov2013clustering}. It was first called this way in \cite{marx2006parameterized}}
	\end{comment}
\end{definition}

Notice that the minimality is important here, otherwise according to the partial ordering definition, one can keep adding extra vertices to the left of $S'$ (towards $X$) and artificially make it more to the left. In order to avoid this, we require all separators we work with to be minimal unless specified otherwise. 
%%Minimum size separators are by default minimal but the other direction does not hold necessarily. 

The notion of leftmost separator is closely related to the notion of \emph{important separator}. Important separator has been defined in \cite{marx2006parameterized}, and then used in \cite{chitnis2013fixed} and \cite{lokshtanov2013clustering}.

The difference between a leftmost separator and an important separator comes from their corresponding partial orders. The partial order defined for important separators is as follows:

\begin{definition}\textbf{Partial Ordering used for Important Separators}. Separator $S \in S_{X,Y}^G$ dominates (or ``is more important than'') separator $S' \in S_{X,Y}^G$ if $|S| \leq |S'|$ and $V_{S,Y} \subset V_{S', Y}$.
\end{definition}
\begin{definition}
	Separator $S \in S_{X,Y}^G$ is an \emph{important ($X$, $Y$, $\leq k$)$^G$-separator} if there exists no other minimal $(X$, $Y$, $\leq k)^G$-separator $S'$ dominating $S$.
\end{definition}

As you see, when ordering important separators we also look at the relation between the sizes but in a leftmost separator, its size just has to be $\leq k$.

\begin{lemma}\label{lem:left-important}
    Every leftmost $(X$, $Y$, $\leq k)^G$-separator is an important $(X$, $Y$, $\leq k)^G$-separator, but the converse does not hold.
\end{lemma}
\begin{proof}
    Let $S$ be a leftmost $(X$, $Y$, $\leq k)^G$-separator. Assume for the sake of contradiction that $S$ is not important. Hence, there exists important $(X$, $Y$, $\leq k)^G$-separator $S'$ that dominates $S$. Therefore $V_{S',Y}\supset V_{S, Y}$. Notice that $S\setminus S'$ is nonempty. Let $v\in S\setminus S'$. Then, $v\in V_{S', Y}$. Therefore, $V_{X, S'}\subset V_{X,S}$. This implies that $S'\preceq S$, which is a contradiction. 
    
    To show that the converse of the proposition does not hold, consider the following counter-example.
    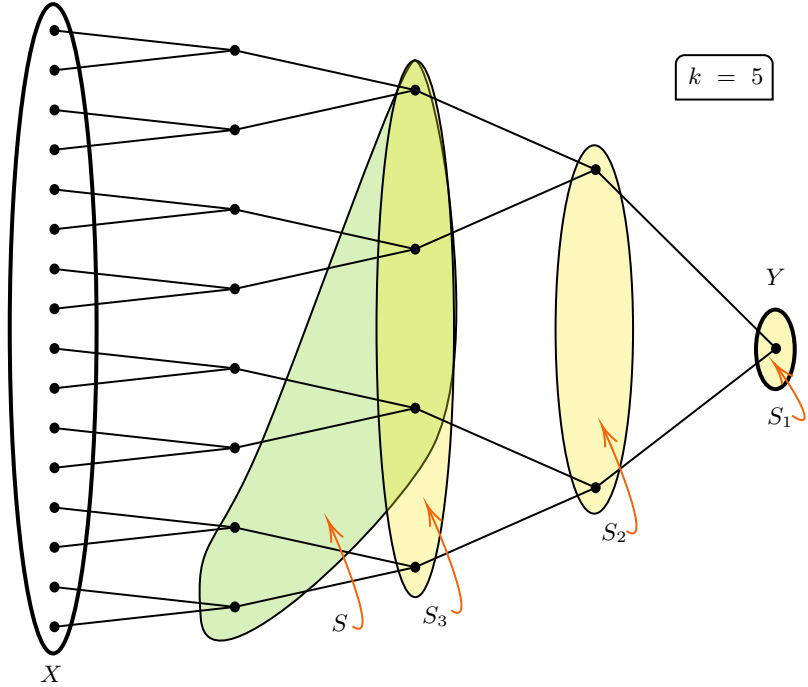
\begin{figure}
        \centering

\tikzset{every picture/.style={line width=0.75pt}} %set default line width to 0.75pt        

\begin{tikzpicture}[x=0.75pt,y=0.75pt,yscale=-1,xscale=1]
%uncomment if require: \path (0,483); %set diagram left start at 0, and has height of 483

%Shape: Polygon Curved [id:ds8709159074994672] 
\draw  [fill={rgb, 255:red, 126; green, 211; blue, 33 }  ,fill opacity=0.3 ] (366,47.33) .. controls (372.26,49.84) and (380.33,92) .. (383,116.67) .. controls (385.67,141.33) and (392.33,205.33) .. (375,241.33) .. controls (357.67,277.33) and (262.33,371.33) .. (259,328) .. controls (255.67,284.67) and (265.93,297.71) .. (291,235.33) .. controls (316.07,172.95) and (359.74,44.83) .. (366,47.33) -- cycle ;
%Shape: Ellipse [id:dp9002692726696586] 
\draw  [fill={rgb, 255:red, 248; green, 231; blue, 28 }  ,fill opacity=0.3 ][line width=1.5]  (536,192.67) .. controls (536,181.62) and (540.33,172.67) .. (545.67,172.67) .. controls (551.01,172.67) and (555.33,181.62) .. (555.33,192.67) .. controls (555.33,203.71) and (551.01,212.67) .. (545.67,212.67) .. controls (540.33,212.67) and (536,203.71) .. (536,192.67) -- cycle ;
%Shape: Ellipse [id:dp5583924722437763] 
\draw  [fill={rgb, 255:red, 248; green, 231; blue, 28 }  ,fill opacity=0.3 ] (346.67,182.33) .. controls (346.67,107.77) and (355.32,47.33) .. (366,47.33) .. controls (376.68,47.33) and (385.33,107.77) .. (385.33,182.33) .. controls (385.33,256.89) and (376.68,317.33) .. (366,317.33) .. controls (355.32,317.33) and (346.67,256.89) .. (346.67,182.33) -- cycle ;
%Shape: Ellipse [id:dp036623925553538994] 
\draw  [fill={rgb, 255:red, 248; green, 231; blue, 28 }  ,fill opacity=0.3 ] (436,182.67) .. controls (436,131.49) and (444.66,90) .. (455.33,90) .. controls (466.01,90) and (474.67,131.49) .. (474.67,182.67) .. controls (474.67,233.85) and (466.01,275.33) .. (455.33,275.33) .. controls (444.66,275.33) and (436,233.85) .. (436,182.67) -- cycle ;
%Shape: Ellipse [id:dp6942153478571038] 
\draw  [fill={rgb, 255:red, 0; green, 0; blue, 0 }  ,fill opacity=1 ] (184,32.25) .. controls (184,31.01) and (184.9,30) .. (186,30) .. controls (187.1,30) and (188,31.01) .. (188,32.25) .. controls (188,33.49) and (187.1,34.5) .. (186,34.5) .. controls (184.9,34.5) and (184,33.49) .. (184,32.25) -- cycle ;
%Shape: Ellipse [id:dp39486117965276124] 
\draw  [fill={rgb, 255:red, 0; green, 0; blue, 0 }  ,fill opacity=1 ] (184,52.25) .. controls (184,51.01) and (184.9,50) .. (186,50) .. controls (187.1,50) and (188,51.01) .. (188,52.25) .. controls (188,53.49) and (187.1,54.5) .. (186,54.5) .. controls (184.9,54.5) and (184,53.49) .. (184,52.25) -- cycle ;
%Shape: Ellipse [id:dp2388285910314496] 
\draw  [fill={rgb, 255:red, 0; green, 0; blue, 0 }  ,fill opacity=1 ] (184,72.25) .. controls (184,71.01) and (184.9,70) .. (186,70) .. controls (187.1,70) and (188,71.01) .. (188,72.25) .. controls (188,73.49) and (187.1,74.5) .. (186,74.5) .. controls (184.9,74.5) and (184,73.49) .. (184,72.25) -- cycle ;
%Shape: Ellipse [id:dp7256047156099972] 
\draw  [fill={rgb, 255:red, 0; green, 0; blue, 0 }  ,fill opacity=1 ] (184,92.25) .. controls (184,91.01) and (184.9,90) .. (186,90) .. controls (187.1,90) and (188,91.01) .. (188,92.25) .. controls (188,93.49) and (187.1,94.5) .. (186,94.5) .. controls (184.9,94.5) and (184,93.49) .. (184,92.25) -- cycle ;
%Shape: Ellipse [id:dp44932749739821487] 
\draw  [fill={rgb, 255:red, 0; green, 0; blue, 0 }  ,fill opacity=1 ] (184,112.25) .. controls (184,111.01) and (184.9,110) .. (186,110) .. controls (187.1,110) and (188,111.01) .. (188,112.25) .. controls (188,113.49) and (187.1,114.5) .. (186,114.5) .. controls (184.9,114.5) and (184,113.49) .. (184,112.25) -- cycle ;
%Shape: Ellipse [id:dp6537741365655658] 
\draw  [fill={rgb, 255:red, 0; green, 0; blue, 0 }  ,fill opacity=1 ] (184,132.25) .. controls (184,131.01) and (184.9,130) .. (186,130) .. controls (187.1,130) and (188,131.01) .. (188,132.25) .. controls (188,133.49) and (187.1,134.5) .. (186,134.5) .. controls (184.9,134.5) and (184,133.49) .. (184,132.25) -- cycle ;
%Shape: Ellipse [id:dp4392247964496989] 
\draw  [fill={rgb, 255:red, 0; green, 0; blue, 0 }  ,fill opacity=1 ] (184,152.25) .. controls (184,151.01) and (184.9,150) .. (186,150) .. controls (187.1,150) and (188,151.01) .. (188,152.25) .. controls (188,153.49) and (187.1,154.5) .. (186,154.5) .. controls (184.9,154.5) and (184,153.49) .. (184,152.25) -- cycle ;
%Shape: Ellipse [id:dp46177382818507606] 
\draw  [fill={rgb, 255:red, 0; green, 0; blue, 0 }  ,fill opacity=1 ] (184,172.25) .. controls (184,171.01) and (184.9,170) .. (186,170) .. controls (187.1,170) and (188,171.01) .. (188,172.25) .. controls (188,173.49) and (187.1,174.5) .. (186,174.5) .. controls (184.9,174.5) and (184,173.49) .. (184,172.25) -- cycle ;
%Shape: Ellipse [id:dp6789840466269836] 
\draw  [fill={rgb, 255:red, 0; green, 0; blue, 0 }  ,fill opacity=1 ] (274,42.25) .. controls (274,41.01) and (274.9,40) .. (276,40) .. controls (277.1,40) and (278,41.01) .. (278,42.25) .. controls (278,43.49) and (277.1,44.5) .. (276,44.5) .. controls (274.9,44.5) and (274,43.49) .. (274,42.25) -- cycle ;
%Shape: Ellipse [id:dp587640706324275] 
\draw  [fill={rgb, 255:red, 0; green, 0; blue, 0 }  ,fill opacity=1 ] (274,82.25) .. controls (274,81.01) and (274.9,80) .. (276,80) .. controls (277.1,80) and (278,81.01) .. (278,82.25) .. controls (278,83.49) and (277.1,84.5) .. (276,84.5) .. controls (274.9,84.5) and (274,83.49) .. (274,82.25) -- cycle ;
%Shape: Ellipse [id:dp18051585797821357] 
\draw  [fill={rgb, 255:red, 0; green, 0; blue, 0 }  ,fill opacity=1 ] (274,122.25) .. controls (274,121.01) and (274.9,120) .. (276,120) .. controls (277.1,120) and (278,121.01) .. (278,122.25) .. controls (278,123.49) and (277.1,124.5) .. (276,124.5) .. controls (274.9,124.5) and (274,123.49) .. (274,122.25) -- cycle ;
%Shape: Ellipse [id:dp8832938869082636] 
\draw  [fill={rgb, 255:red, 0; green, 0; blue, 0 }  ,fill opacity=1 ] (274,162.25) .. controls (274,161.01) and (274.9,160) .. (276,160) .. controls (277.1,160) and (278,161.01) .. (278,162.25) .. controls (278,163.49) and (277.1,164.5) .. (276,164.5) .. controls (274.9,164.5) and (274,163.49) .. (274,162.25) -- cycle ;
%Shape: Ellipse [id:dp5866696864357572] 
\draw  [fill={rgb, 255:red, 0; green, 0; blue, 0 }  ,fill opacity=1 ] (364,62.25) .. controls (364,61.01) and (364.9,60) .. (366,60) .. controls (367.1,60) and (368,61.01) .. (368,62.25) .. controls (368,63.49) and (367.1,64.5) .. (366,64.5) .. controls (364.9,64.5) and (364,63.49) .. (364,62.25) -- cycle ;
%Shape: Ellipse [id:dp4575604886482716] 
\draw  [fill={rgb, 255:red, 0; green, 0; blue, 0 }  ,fill opacity=1 ] (364,142.25) .. controls (364,141.01) and (364.9,140) .. (366,140) .. controls (367.1,140) and (368,141.01) .. (368,142.25) .. controls (368,143.49) and (367.1,144.5) .. (366,144.5) .. controls (364.9,144.5) and (364,143.49) .. (364,142.25) -- cycle ;
%Shape: Ellipse [id:dp8808122469699551] 
\draw  [fill={rgb, 255:red, 0; green, 0; blue, 0 }  ,fill opacity=1 ] (454,102.25) .. controls (454,101.01) and (454.9,100) .. (456,100) .. controls (457.1,100) and (458,101.01) .. (458,102.25) .. controls (458,103.49) and (457.1,104.5) .. (456,104.5) .. controls (454.9,104.5) and (454,103.49) .. (454,102.25) -- cycle ;
%Shape: Ellipse [id:dp8800768070714231] 
\draw  [fill={rgb, 255:red, 0; green, 0; blue, 0 }  ,fill opacity=1 ] (184,192.25) .. controls (184,191.01) and (184.9,190) .. (186,190) .. controls (187.1,190) and (188,191.01) .. (188,192.25) .. controls (188,193.49) and (187.1,194.5) .. (186,194.5) .. controls (184.9,194.5) and (184,193.49) .. (184,192.25) -- cycle ;
%Shape: Ellipse [id:dp0666339889387253] 
\draw  [fill={rgb, 255:red, 0; green, 0; blue, 0 }  ,fill opacity=1 ] (184,212.25) .. controls (184,211.01) and (184.9,210) .. (186,210) .. controls (187.1,210) and (188,211.01) .. (188,212.25) .. controls (188,213.49) and (187.1,214.5) .. (186,214.5) .. controls (184.9,214.5) and (184,213.49) .. (184,212.25) -- cycle ;
%Shape: Ellipse [id:dp9652167592909129] 
\draw  [fill={rgb, 255:red, 0; green, 0; blue, 0 }  ,fill opacity=1 ] (184,232.25) .. controls (184,231.01) and (184.9,230) .. (186,230) .. controls (187.1,230) and (188,231.01) .. (188,232.25) .. controls (188,233.49) and (187.1,234.5) .. (186,234.5) .. controls (184.9,234.5) and (184,233.49) .. (184,232.25) -- cycle ;
%Shape: Ellipse [id:dp18011991829983076] 
\draw  [fill={rgb, 255:red, 0; green, 0; blue, 0 }  ,fill opacity=1 ] (184,252.25) .. controls (184,251.01) and (184.9,250) .. (186,250) .. controls (187.1,250) and (188,251.01) .. (188,252.25) .. controls (188,253.49) and (187.1,254.5) .. (186,254.5) .. controls (184.9,254.5) and (184,253.49) .. (184,252.25) -- cycle ;
%Shape: Ellipse [id:dp5640028234560648] 
\draw  [fill={rgb, 255:red, 0; green, 0; blue, 0 }  ,fill opacity=1 ] (184,272.25) .. controls (184,271.01) and (184.9,270) .. (186,270) .. controls (187.1,270) and (188,271.01) .. (188,272.25) .. controls (188,273.49) and (187.1,274.5) .. (186,274.5) .. controls (184.9,274.5) and (184,273.49) .. (184,272.25) -- cycle ;
%Shape: Ellipse [id:dp05472406022466836] 
\draw  [fill={rgb, 255:red, 0; green, 0; blue, 0 }  ,fill opacity=1 ] (184,292.25) .. controls (184,291.01) and (184.9,290) .. (186,290) .. controls (187.1,290) and (188,291.01) .. (188,292.25) .. controls (188,293.49) and (187.1,294.5) .. (186,294.5) .. controls (184.9,294.5) and (184,293.49) .. (184,292.25) -- cycle ;
%Shape: Ellipse [id:dp8887583603609357] 
\draw  [fill={rgb, 255:red, 0; green, 0; blue, 0 }  ,fill opacity=1 ] (184,312.25) .. controls (184,311.01) and (184.9,310) .. (186,310) .. controls (187.1,310) and (188,311.01) .. (188,312.25) .. controls (188,313.49) and (187.1,314.5) .. (186,314.5) .. controls (184.9,314.5) and (184,313.49) .. (184,312.25) -- cycle ;
%Shape: Ellipse [id:dp9742312528551429] 
\draw  [fill={rgb, 255:red, 0; green, 0; blue, 0 }  ,fill opacity=1 ] (184,332.25) .. controls (184,331.01) and (184.9,330) .. (186,330) .. controls (187.1,330) and (188,331.01) .. (188,332.25) .. controls (188,333.49) and (187.1,334.5) .. (186,334.5) .. controls (184.9,334.5) and (184,333.49) .. (184,332.25) -- cycle ;
%Shape: Ellipse [id:dp6134054101616795] 
\draw  [fill={rgb, 255:red, 0; green, 0; blue, 0 }  ,fill opacity=1 ] (274,202.25) .. controls (274,201.01) and (274.9,200) .. (276,200) .. controls (277.1,200) and (278,201.01) .. (278,202.25) .. controls (278,203.49) and (277.1,204.5) .. (276,204.5) .. controls (274.9,204.5) and (274,203.49) .. (274,202.25) -- cycle ;
%Shape: Ellipse [id:dp977254702830493] 
\draw  [fill={rgb, 255:red, 0; green, 0; blue, 0 }  ,fill opacity=1 ] (274,242.25) .. controls (274,241.01) and (274.9,240) .. (276,240) .. controls (277.1,240) and (278,241.01) .. (278,242.25) .. controls (278,243.49) and (277.1,244.5) .. (276,244.5) .. controls (274.9,244.5) and (274,243.49) .. (274,242.25) -- cycle ;
%Shape: Ellipse [id:dp8460355697616597] 
\draw  [fill={rgb, 255:red, 0; green, 0; blue, 0 }  ,fill opacity=1 ] (274,282.25) .. controls (274,281.01) and (274.9,280) .. (276,280) .. controls (277.1,280) and (278,281.01) .. (278,282.25) .. controls (278,283.49) and (277.1,284.5) .. (276,284.5) .. controls (274.9,284.5) and (274,283.49) .. (274,282.25) -- cycle ;
%Shape: Ellipse [id:dp3940853208082702] 
\draw  [fill={rgb, 255:red, 0; green, 0; blue, 0 }  ,fill opacity=1 ] (274,322.25) .. controls (274,321.01) and (274.9,320) .. (276,320) .. controls (277.1,320) and (278,321.01) .. (278,322.25) .. controls (278,323.49) and (277.1,324.5) .. (276,324.5) .. controls (274.9,324.5) and (274,323.49) .. (274,322.25) -- cycle ;
%Shape: Ellipse [id:dp7117000866104188] 
\draw  [fill={rgb, 255:red, 0; green, 0; blue, 0 }  ,fill opacity=1 ] (364,222.25) .. controls (364,221.01) and (364.9,220) .. (366,220) .. controls (367.1,220) and (368,221.01) .. (368,222.25) .. controls (368,223.49) and (367.1,224.5) .. (366,224.5) .. controls (364.9,224.5) and (364,223.49) .. (364,222.25) -- cycle ;
%Shape: Ellipse [id:dp17973649353439303] 
\draw  [fill={rgb, 255:red, 0; green, 0; blue, 0 }  ,fill opacity=1 ] (364,302.25) .. controls (364,301.01) and (364.9,300) .. (366,300) .. controls (367.1,300) and (368,301.01) .. (368,302.25) .. controls (368,303.49) and (367.1,304.5) .. (366,304.5) .. controls (364.9,304.5) and (364,303.49) .. (364,302.25) -- cycle ;
%Shape: Ellipse [id:dp2107210968448685] 
\draw  [fill={rgb, 255:red, 0; green, 0; blue, 0 }  ,fill opacity=1 ] (454,262.25) .. controls (454,261.01) and (454.9,260) .. (456,260) .. controls (457.1,260) and (458,261.01) .. (458,262.25) .. controls (458,263.49) and (457.1,264.5) .. (456,264.5) .. controls (454.9,264.5) and (454,263.49) .. (454,262.25) -- cycle ;
%Shape: Ellipse [id:dp42077836855421946] 
\draw  [fill={rgb, 255:red, 0; green, 0; blue, 0 }  ,fill opacity=1 ] (544,192.25) .. controls (544,191.01) and (544.9,190) .. (546,190) .. controls (547.1,190) and (548,191.01) .. (548,192.25) .. controls (548,193.49) and (547.1,194.5) .. (546,194.5) .. controls (544.9,194.5) and (544,193.49) .. (544,192.25) -- cycle ;
%Straight Lines [id:da13963249682304846] 
\draw    (456,102.25) -- (546,192.25) ;
%Straight Lines [id:da615134684432769] 
\draw    (366,62.25) -- (456,102.25) ;
%Straight Lines [id:da36338403379251427] 
\draw    (366,142.25) -- (456,102.25) ;
%Straight Lines [id:da9608816323257021] 
\draw    (366,222.25) -- (456,262.25) ;
%Straight Lines [id:da9112002885161734] 
\draw    (456,262.25) -- (546,192.25) ;
%Straight Lines [id:da9593783901542141] 
\draw    (366,302.25) -- (456,262.25) ;
%Straight Lines [id:da2636109251322507] 
\draw    (276,42.25) -- (366,62.25) ;
%Straight Lines [id:da08035456230763316] 
\draw    (276,82.25) -- (366,62.25) ;
%Straight Lines [id:da2652237080944966] 
\draw    (276,122.25) -- (366,142.25) ;
%Straight Lines [id:da31941681020176715] 
\draw    (276,162.25) -- (366,142.25) ;
%Straight Lines [id:da6267078162289763] 
\draw    (276,202.25) -- (366,222.25) ;
%Straight Lines [id:da2899624346842269] 
\draw    (276,242.25) -- (366,222.25) ;
%Straight Lines [id:da03887272025603328] 
\draw    (276,282.25) -- (366,302.25) ;
%Straight Lines [id:da835329615449991] 
\draw    (276,322.25) -- (366,302.25) ;
%Straight Lines [id:da9190262578087909] 
\draw    (186,32.25) -- (276,42.25) ;
%Straight Lines [id:da894563484476044] 
\draw    (186,52.25) -- (276,42.25) ;
%Straight Lines [id:da7628908605119831] 
\draw    (186,72.25) -- (276,82.25) ;
%Straight Lines [id:da013222005702398842] 
\draw    (186,92.25) -- (276,82.25) ;
%Straight Lines [id:da5537460538644672] 
\draw    (186,112.25) -- (276,122.25) ;
%Straight Lines [id:da7678298826798273] 
\draw    (186,132.25) -- (276,122.25) ;
%Straight Lines [id:da24222760190556358] 
\draw    (186,152.25) -- (276,162.25) ;
%Straight Lines [id:da6216256540814542] 
\draw    (186,172.25) -- (276,162.25) ;
%Straight Lines [id:da04718940407684835] 
\draw    (186,192.25) -- (276,202.25) ;
%Straight Lines [id:da3768637859169832] 
\draw    (186,212.25) -- (276,202.25) ;
%Straight Lines [id:da3825777455888615] 
\draw    (186,232.25) -- (276,242.25) ;
%Straight Lines [id:da25609772554385923] 
\draw    (186,252.25) -- (276,242.25) ;
%Straight Lines [id:da4038147911408718] 
\draw    (186,272.25) -- (276,282.25) ;
%Straight Lines [id:da06848654622311567] 
\draw    (186,292.25) -- (276,282.25) ;
%Straight Lines [id:da3243128387398886] 
\draw    (186,312.25) -- (276,322.25) ;
%Straight Lines [id:da6166062112400741] 
\draw    (186,332.25) -- (276,322.25) ;
%Shape: Ellipse [id:dp6816034708164131] 
\draw  [line width=1.5]  (163.67,182.33) .. controls (163.67,92.13) and (173.37,19) .. (185.33,19) .. controls (197.3,19) and (207,92.13) .. (207,182.33) .. controls (207,272.54) and (197.3,345.67) .. (185.33,345.67) .. controls (173.37,345.67) and (163.67,272.54) .. (163.67,182.33) -- cycle ;
%Curve Lines [id:da7707543155339924] 
\draw [color={rgb, 255:red, 242; green, 99; blue, 10 }  ,draw opacity=1 ]   (554,226) .. controls (567.76,235.28) and (556.02,214.15) .. (546.21,201.24) ;
\draw [shift={(545,199.67)}, rotate = 411.71000000000004] [color={rgb, 255:red, 242; green, 99; blue, 10 }  ,draw opacity=1 ][line width=0.75]    (10.93,-3.29) .. controls (6.95,-1.4) and (3.31,-0.3) .. (0,0) .. controls (3.31,0.3) and (6.95,1.4) .. (10.93,3.29)   ;
%Curve Lines [id:da6847155918309036] 
\draw [color={rgb, 255:red, 242; green, 99; blue, 10 }  ,draw opacity=1 ]   (471.33,284) .. controls (485.17,293.33) and (467.02,248.02) .. (459.75,232.26) ;
\draw [shift={(459,230.67)}, rotate = 424.53999999999996] [color={rgb, 255:red, 242; green, 99; blue, 10 }  ,draw opacity=1 ][line width=0.75]    (10.93,-3.29) .. controls (6.95,-1.4) and (3.31,-0.3) .. (0,0) .. controls (3.31,0.3) and (6.95,1.4) .. (10.93,3.29)   ;
%Curve Lines [id:da6204184567005024] 
\draw [color={rgb, 255:red, 242; green, 99; blue, 10 }  ,draw opacity=1 ]   (384,324) .. controls (397.83,333.33) and (379.69,288.02) .. (372.41,272.26) ;
\draw [shift={(371.67,270.67)}, rotate = 424.53999999999996] [color={rgb, 255:red, 242; green, 99; blue, 10 }  ,draw opacity=1 ][line width=0.75]    (10.93,-3.29) .. controls (6.95,-1.4) and (3.31,-0.3) .. (0,0) .. controls (3.31,0.3) and (6.95,1.4) .. (10.93,3.29)   ;
%Curve Lines [id:da2892290958643884] 
\draw [color={rgb, 255:red, 242; green, 99; blue, 10 }  ,draw opacity=1 ]   (334.67,332) .. controls (348.5,341.33) and (330.36,296.02) .. (323.08,280.26) ;
\draw [shift={(322.33,278.67)}, rotate = 424.53999999999996] [color={rgb, 255:red, 242; green, 99; blue, 10 }  ,draw opacity=1 ][line width=0.75]    (10.93,-3.29) .. controls (6.95,-1.4) and (3.31,-0.3) .. (0,0) .. controls (3.31,0.3) and (6.95,1.4) .. (10.93,3.29)   ;
%Rounded Same Side Corner Rect [id:dp08680998586091837] 
\draw   (496,49.13) .. controls (496,46.67) and (498,44.67) .. (500.47,44.67) -- (539.87,44.67) .. controls (542.33,44.67) and (544.33,46.67) .. (544.33,49.13) -- (544.33,67) .. controls (544.33,67) and (544.33,67) .. (544.33,67) -- (496,67) .. controls (496,67) and (496,67) .. (496,67) -- cycle ;

% Text Node
\draw (500.67,49.4) node [anchor=north west][inner sep=0.75pt]  [font=\footnotesize]  {$k\ =\ 5$};
% Text Node
\draw (539.33,150.4) node [anchor=north west][inner sep=0.75pt]    {$Y$};
% Text Node
\draw (177.33,350.4) node [anchor=north west][inner sep=0.75pt]    {$X$};
% Text Node
\draw (322.67,324.73) node [anchor=north west][inner sep=0.75pt]  [font=\small]  {$S$};
% Text Node
\draw (368,320.73) node [anchor=north west][inner sep=0.75pt]  [font=\small]  {$S_{3}$};
% Text Node
\draw (457.33,278.73) node [anchor=north west][inner sep=0.75pt]  [font=\small]  {$S_{2}$};
% Text Node
\draw (540,220.73) node [anchor=north west][inner sep=0.75pt]  [font=\small]  {$S_{1}$};

\end{tikzpicture}
        \caption{Here $k = 5$. Separators $S_1$, $S_2$, $S_3$, and $S$ are all important but only $S$ is leftmost}
        \label{fig:counter-ex}
    \end{figure}
\end{proof}
 
 As shown in Figure~\ref{fig:counter-ex}, not all the important separators are leftmost. Our purpose is to find a separator more towards the bigger side in order to have more balanced separators. For that reason, not all the important separators are good. For instance we do not need to consider $S_1$ because that is the most unbalanced separator one can find. This is the main reason that we defined the new notion of leftmost separators. The reason that $S_1$ in Figure~\ref{fig:counter-ex} is important is that $|S_1| \leq |S_2|$ (even though $V_{S_2,Y}$ contains $V_{S_1, Y}$). 
 As argued, leftmost separators are better candidates for our application. However, as the reader can see, there is a strong similarity between these two notions. We give tight upper bounds for the number of leftmost separators and a tight upper bound for the number of important separators.
 
 \begin{lemma}\label{lem:numleft_numimportant}
     Let $\mathcal{A}_{X,Y, \leq k}^{G}$ and $\mathcal{B}_{X,Y, \leq k}^{G}$ be the set of all leftmost ($X, Y, \leq k$)$^G$-separators and the set of all important ($X, Y, \leq k$)$^G$-separators, respectively. Then, 
     \[
     \mathcal{B}_{X,Y, \leq k}^{G} = \bigcup\limits_{i=1}^{k}\mathcal{A}_{X,Y, \leq i}^{G}
     \]
 \end{lemma}
 \begin{proof}
     Based on Lemma~\ref{lem:left-important}, for  every $1\leq i\leq k$, every leftmost ($X, Y, \leq i$)$^G$-separator is an important ($X, Y, \leq i$)$^G$-separator. Furthermore, every important ($X, Y, \leq i$)$^G$-separator is an important ($X, Y, \leq k$)$^G$-separator. Hence, \[
     \mathcal{B}_{X,Y, \leq k}^{G} \supseteq \bigcup\limits_{i=1}^{k}\mathcal{A}_{X,Y, \leq i}^{G}\]
     
     For the other side, we show that every important ($X, Y, \leq k$)$^G$-separator $S$ is a leftmost ($X, Y, \leq |S|$)$^G$-separator.
     For the sake of contradiction, assume that there exist important ($X, Y, \leq k$)$^G$-separator $S$ that is not a leftmost ($X, Y, \leq |S|$)$^G$-separator (Notice that $\mathcal{A}_{X,Y, \leq |S|}^{G}$ is nonempty since there is at least one separator of $\leq |S|$ between $X$ and $Y$). This means there exist leftmost ($X, Y, \leq |S|$)$^G$-separator $S'$ that $S' \preceq S$. Notice that $|S'| \leq |S|$ and also $V_{S', Y} \supset V_{S, Y}$. This means $S'$ is dominates $S$ and hence $S$ is not important, which is a contradiction.
 \end{proof}

In this paper, we show that the number of leftmost ($X$, $Y$, $\leq k$)-separators is $\leq C_{k-1}$, where $C_n$ is the $n$-th Catalan number. Furthermore, we close the gap and show that this upper bound is tight. Then, we give an $\mathcal{O}(4^{k}k n)$-time algorithm finding all minimal leftmost ($X$, $Y$, $\leq k$)-separators. Notice that $C_{k-1}\sim\frac{4^{k-1}}{\sqrt{\pi}(k-1)^{\frac{3}{2}}}$.

Based on Lemma~\ref{lem:numleft_numimporta}, this implies that the number of important ($X$, $Y$, $\leq k$)-separators is $\leq \sum\limits_{i=1}^{k-1}C_{i}$ and the bound is tight.

One of the important applications of the algorithm that finds all the leftmost separators is treewidth approximation. Treewidth approximation is a crucial problem in computer science. Courcelle's methatheorem \cite{courcelle1990monadic} states that every problem which can be described in monadic second order logic has an FPT algorithm with the treewidth $k$ as its parameter. An FPT algorithm is an algorithm that runs in time $\mathcal{O}(f(k)\, poly(n))$, where $n$ is the input size, $k$ is the parameter (here, treewidth), and $f(\cdot)$ is a computable function.

So, based on Courcelle's methatheorem, many NP-complete graph problems obtain polynomial algorithms (in terms of the input size), and hence they can be solved fast if the treewidth is small. These algorithms require access to tree decompositions of small width. However, finding the exact treewidth itself is another NP-complete problem \cite{arnborg1987complexity}. Here, we look for an approximation algorithm to solve the treewidth problem.

\begin{tcolorbox}
	\textbf{Problem 0.} Given a graph $G = (V, E)$, and an integer $k \in \mathbb{N}$, is the treewidth of $G$ at most $k$? If yes, output a tree decomposition with width $\leq \alpha k$, where $\alpha \geq 1$ is a constant. Otherwise, output a subgraph which is the bottleneck.
\end{tcolorbox}

There are various algorithms solving this problem for different $\alpha$'s (the approximation ratio). As mentioned above, we are interested in constant-factor approximation FPT algorithms. Table \ref{tab:history} shows a brief history of the previous work.
\begin{table}
\begin{center}
\resizebox{\columnwidth}{!}{%
	\begin{tabular}{|| c  | c |  c | c | c ||}
		\hline
		Reference & Approximation& Dependence on $k$ & Dependence on $n$ & Comments\\& Ratio $\alpha$ &&& \\ 
		\hline\hline
		Robertson \& Seymour \cite{robertson1995graph}& $4$ & $ 
		%%% to be consistent with other entries \mathcal{O}(
		3^{3k}$ &$n^2$ &\\
		\hline
		Lagergren \cite{lagergren1996efficient} & $8$ & $2^{\mathcal{O}(k\log k)}$ &$n \log^2 n$ &\\
		\hline
		Reed \cite{reed1992finding}& $7$ or $5$ & $2^{24k}k!$ & $n \log n$ & Large coefficient \\ &&&&of $k$\\&&&&in the exponent\\
		\hline
		Amir \cite{amir2010approximation}& $4.5$ & $2^{3k}$ & $n^2$ & \\
		\hline
		Amir \cite{amir2010approximation} & $\frac{11}{3}$ & $2^{3.6982 k}$ & $n^2$&\\
		\hline
		Bodlaender et al.\ \cite{bodlaender2016c}& $3$ & $2^{\mathcal{O}(k)}$ & $n \log n$& The coefficient\\&&&&of $k$\\&&&&is not stated\\
		\hline
		Bodlaender et al.\ \cite{bodlaender2016c}& $5$ & $2^{\mathcal{O}(k)}$  & $n$ & The coefficient\\&&&&of $k$\\&&&&is not stated\\
		\hline
		Belbasi \& F\"urer \cite{belbasi2020improvement}& $5$ & $2^{8.766k}$ & $n \log n$ &Small \\&&&&coefficient of $k$\\
		\hline
		Korhonen \cite{korhonen2021single} & 2 & $2^{\mathcal{O}(k)}$&$n$&better coefficient\\&&&& of $k$\\&&&&compared to \cite{bodlaender2016c}\\&&&& but still not\\&&&&very applicable.\\
		\hline
		This paper & $5$ & $2^{6.755k}$ & $n \log n$ & Practical \\&&&& for small $k$\\
		\hline
	\end{tabular}
	}
\end{center}\label{tab:history}
\end{table}

Algorithms \cite{bodlaender2016c} and \cite{korhonen2021single} both run in $2^{\mathcal{O}(k)}n$ time, which is linear in $n$. However, the coefficients of $k$ in the exponent are large. The former one does not mention the exact  coefficient and seems to have a very large coefficient. The latter one, which is a very recent paper, mentions that the coefficient of $k$ in the exponent is some number between 10 and 11. Our goal is to make treewidth approximation more applicable by decreasing the coefficient of $k$ in the exponent. We can afford an extra $\log n$ factor in the running time in order to reduce the huge dependence on $k$. We sacrifice the linear dependence on $n$, and give an algorithm which runs way faster in various cases. So, let us look at $n\log n$-time algorithms. Reed \cite{reed1992finding} gave the first $n \log n$-time algorithm. He did not mention the dependence on $k$ precisely but a detailed analysis in \cite{belbasi2020improvement} shows that it is $\mathcal{O}(2^{24k}k!)$. Here, even though $2^{24k} = o(k!)$, actually $k!$ is reasonable for small $k$'s while $2^{24k}$ is not. Later on, the authors of this paper introduced an $\mathcal{O}(2^{8.766k} n \log n)$-algorithm \cite{belbasi2020improvement}. The algorithm presented in this paper is based on \cite{reed1992finding} and \cite{belbasi2020improvement}. In these papers, when it is known that a good separator $S$ exists between $X$ and $Y$, an efficient algorithm finds an arbitrary separator between $X$ and $Y$. The ability to find leftmost separators allows for an improvement. If $S$ is a good separator between $X$ and $Y$, and $V_{X,S}$ is estimated to be at least as big as $V_{S,Y}$, then the best leftmost separator between $X$ and $Y$ has a definite advantage.

Instead of a balanced separator with minimum size, we consider all leftmost separators (closest possible to the bigger side). This helps us to obtain an $\mathcal{O}(2^{6.755k} n\log n)$-time algorithm with the same approximation ratio of $5$ as in \cite{reed1992finding} and \cite{belbasi2020improvement}. 

Before moving onto the next section, we have to mention that the algorithm to find all leftmost separators works for both directed and undirected graphs.

Below, we summarize our contributions.

\subsection{Our Contributions} 

First, we give a tight upper bound on the number of the leftmost separators.

\begin{theorem}
	Let $G = (V, E)$ be a graph, $X, Y \subseteq V$, and $k \in \mathbb{N}$. The number of leftmost $(X$, $Y$, $\leq k)^G$-separators\footnote{Notice that all leftmost separators are minimal per definition.} is at most $C_{k-1}$\footnote{$C_n$ is the $n$th Catalan number}$= \frac{1}{k}{2(k-1) \choose k-1} \sim \frac{4^{k-1}}{\sqrt{\pi}(k-1)^{3/2}}$. Furthermore, the number of important $(X$, $Y$, $\leq k)^G$-separators is at most $\sum\limits_{i=0}^{k-1}C_i$. Both bounds are tight.
\end{theorem}
Then, we give an algorithm finding all leftmost separators.
\begin{theorem}
	Let $G = (V, E)$ be a graph, $X, Y \subseteq V$, and $k \in \mathbb{N}$. There is an $\mathcal{O}(2^{2k}\sqrt{k}n)$-time algorithm which outputs all the leftmost $(X$, $Y$, $\leq k)^G$-separators.
\end{theorem}
Now, we use the algorithm finding all the leftmost separators to solve treewidth approximation much faster.
\begin{theorem}
	Let $G = (V, E)$ be a graph, and $k \in \mathbb{N}$. There is an algorithm that either outputs a tree decomposition of $G$ with width $\leq 5(k-1)$, or determines that $tw(G) > k-1$ in time $\mathcal{O}\left(2^{6.755k}n \log n\right)$.
\end{theorem}

\section{Finding the Leftmost Minimum Size  ($X$, $Y$, $\leq k$)$^G$-Separator}\label{sec:flumgrohe}

In this section, we review an algorithm for the following problem which has been fully described in \cite{FlumGrohe} in details. It is based on \cite{menger1927allgemeinen}, \cite{robertson1995graph}, and \cite{reed1992finding}.\\
\begin{tcolorbox}
	\textbf{Problem 1.} Given a graph $G = (V, E)$, sets $X, Y \subseteq V$, and $k \in \mathbb{N}$, is there an ($X$, $Y$, $\leq k$)-separator?
\end{tcolorbox}
\begin{lemma}{(Lemma 11.20 of \cite{FlumGrohe})} There is an algorithm which solves Problem~1 in time $\mathcal{O}(k |V|)$.
\end{lemma}
The proof has been given in \cite{FlumGrohe}. Here, we just briefly mention the crux of the idea and apply a tiny modification. Later on, we write the pseudocode of Algorithm~\ref{algo:1} since we use it in Algorithm \ref{sec:min_left}. The main idea is based on the following theorem of Menger \cite{menger1927allgemeinen}.
\begin{theorem}(A version of Menger's theorem \cite{menger1927allgemeinen})
	Let $G = (V, E)$ be a graph, and $X, Y \subseteq V$. Then, the size of the minimum  $(X, Y)^G$-separator is equal to the maximal number of disjoint paths from vertices of $X$ to vertices of $Y$.
\end{theorem}
So, the problem reduces to finding the number of disjoint paths from $X$ to $Y$, and this itself reduces to a network flow problem (with the capacities on the vertices and not the edges).

\begin{definition}
	Let $\mathcal{P}$ be a family of disjoint paths from a set $X \subseteq V$ to a set $Y \subseteq V$ in a graph $G= (V, E)$. We call $Q$ a \emph{$\mathcal{P}$-augmenting walk} if\/ $Q = q_1\cdots q_s$ such that $q_i \in V$ for all $i$ with $1\leq i \leq s$, and $\{q_i, q_{i+1}\} \in E$ for all $i$ with $1\leq i \leq s-1$, and also
	\begin{enumerate}
		\item No edge shows up twice on $Q$,
		\item If $Q$ intersects $P = p_1 \cdots p_l \in \mathcal{P}$ at $q_i = p_j$, then $i \neq s$ and $q_{i+1} = p_{j-1}$ (i.e., if $Q$ intersects a path $P \in \mathcal{P}$, they share at least one edge and that edge appears in opposite directions on $P$ and $Q$). 
	\end{enumerate}
\end{definition}
\begin{claim}\label{col:menger}(Claim 1 of Lemma 11.20 in \cite{FlumGrohe}) 
	Let $\mathcal{P}$ be a family of pairwise disjoint paths from $X$ to $Y$ and let $Q$ be a $\mathcal{P}$-augmenting walk from $X$ to $Y$. Then, there exists a family of pairwise disjoint paths from $X$ to $Y$ of size $|\mathcal{P}| + 1$.
\end{claim}

The idea is to think of paths in $\mathcal{P}$ and also $Q$ as sending a unit flow from $X$ to $Y$\footnote{Again, notice that the capacities are on the vertices and not the edges.} (the flows in opposite directions cancel each other  when $Q$ intersects with a path in $\mathcal{P}$). This gives a new family of pairwise disjoint paths with one more disjoint path. Then, we keep trying to find another $\mathcal{P}$-augmenting walk and update $\mathcal{P}$ until it is impossible to proceed (no $\mathcal{P}$-augmenting walk is found). Check Figure~\ref{fig:augmenting}; initially, $\mathcal{P} = \{P_1, P_2, P_3, P_4\}$. We find an $\mathcal{P}$-augmenting path $Q$. After sending a unit flow through all the paths, we get a new set of paths $\mathcal{P}' = \{\textcolor{my_red}{P'_0}, \textcolor{my_green}{P'_1}, \textcolor{my_blue}{P'_2}, \textcolor{my_purple}{P'_3}, \textcolor{my_olive}{P'_4}\}$ of size one more.

\begin{figure}
    \centering

\tikzset{every picture/.style={line width=0.75pt}} %set default line width to 0.75pt        

\begin{tikzpicture}[x=0.35pt,y=0.35pt,yscale=-1,xscale=1]
%uncomment if require: \path (0,561); %set diagram left start at 0, and has height of 561

%Shape: Polygon Curved [id:ds39982806644378877] 
\draw  [draw opacity=0][fill={rgb, 255:red, 245; green, 166; blue, 35 }  ,fill opacity=0.3 ] (62,339.5) .. controls (63,304.5) and (400.81,330.01) .. (416.94,337.05) .. controls (433.07,344.08) and (442.14,352.36) .. (433,361.5) .. controls (423.86,370.64) and (282.68,382.81) .. (282,402.5) .. controls (281.32,422.19) and (173.88,427.49) .. (161,450.5) .. controls (148.12,473.51) and (426.66,451.54) .. (424,467.5) .. controls (421.34,483.46) and (405,490.5) .. (403,499.5) .. controls (401,508.5) and (134,480.5) .. (114,450.5) .. controls (94,420.5) and (229,433.5) .. (240,400.45) .. controls (251,367.4) and (342.75,365.04) .. (384,357.5) .. controls (425.25,349.96) and (61,374.5) .. (62,339.5) -- cycle ;
%Straight Lines [id:da7571262381039565] 
\draw [color={rgb, 255:red, 207; green, 206; blue, 206 }  ,draw opacity=1 ]   (76,379.03) -- (621,340.03) ;
%Straight Lines [id:da36106857350893606] 
\draw [color={rgb, 255:red, 207; green, 206; blue, 206 }  ,draw opacity=1 ]   (76,419.03) -- (621,380.03) ;
%Straight Lines [id:da14682170384127335] 
\draw [color={rgb, 255:red, 207; green, 206; blue, 206 }  ,draw opacity=1 ]   (76,459.03) -- (621,420.03) ;
%Straight Lines [id:da4426461498544034] 
\draw [color={rgb, 255:red, 207; green, 206; blue, 206 }  ,draw opacity=1 ]   (76,499.03) -- (621,460.03) ;
%Rounded Rect [id:dp9185935975761685] 
\draw   (41,50) .. controls (41,42.27) and (47.27,36) .. (55,36) -- (97,36) .. controls (104.73,36) and (111,42.27) .. (111,50) -- (111,239) .. controls (111,246.73) and (104.73,253) .. (97,253) -- (55,253) .. controls (47.27,253) and (41,246.73) .. (41,239) -- cycle ;
%Shape: Circle [id:dp3503313788743512] 
\draw  [fill={rgb, 255:red, 0; green, 0; blue, 0 }  ,fill opacity=1 ] (72,62) .. controls (72,59.79) and (73.79,58) .. (76,58) .. controls (78.21,58) and (80,59.79) .. (80,62) .. controls (80,64.21) and (78.21,66) .. (76,66) .. controls (73.79,66) and (72,64.21) .. (72,62) -- cycle ;
%Shape: Circle [id:dp8430154795208267] 
\draw  [fill={rgb, 255:red, 0; green, 0; blue, 0 }  ,fill opacity=1 ] (72,102) .. controls (72,99.79) and (73.79,98) .. (76,98) .. controls (78.21,98) and (80,99.79) .. (80,102) .. controls (80,104.21) and (78.21,106) .. (76,106) .. controls (73.79,106) and (72,104.21) .. (72,102) -- cycle ;
%Shape: Circle [id:dp45190024273950513] 
\draw  [fill={rgb, 255:red, 0; green, 0; blue, 0 }  ,fill opacity=1 ] (72,142) .. controls (72,139.79) and (73.79,138) .. (76,138) .. controls (78.21,138) and (80,139.79) .. (80,142) .. controls (80,144.21) and (78.21,146) .. (76,146) .. controls (73.79,146) and (72,144.21) .. (72,142) -- cycle ;
%Shape: Circle [id:dp4019062232753654] 
\draw  [fill={rgb, 255:red, 0; green, 0; blue, 0 }  ,fill opacity=1 ] (72,182) .. controls (72,179.79) and (73.79,178) .. (76,178) .. controls (78.21,178) and (80,179.79) .. (80,182) .. controls (80,184.21) and (78.21,186) .. (76,186) .. controls (73.79,186) and (72,184.21) .. (72,182) -- cycle ;
%Shape: Circle [id:dp664858788005956] 
\draw  [fill={rgb, 255:red, 0; green, 0; blue, 0 }  ,fill opacity=1 ] (72,222) .. controls (72,219.79) and (73.79,218) .. (76,218) .. controls (78.21,218) and (80,219.79) .. (80,222) .. controls (80,224.21) and (78.21,226) .. (76,226) .. controls (73.79,226) and (72,224.21) .. (72,222) -- cycle ;

%Rounded Rect [id:dp9533195748297825] 
\draw   (586,51) .. controls (586,43.27) and (592.27,37) .. (600,37) -- (642,37) .. controls (649.73,37) and (656,43.27) .. (656,51) -- (656,240) .. controls (656,247.73) and (649.73,254) .. (642,254) -- (600,254) .. controls (592.27,254) and (586,247.73) .. (586,240) -- cycle ;
%Shape: Circle [id:dp8886597847553477] 
\draw  [fill={rgb, 255:red, 0; green, 0; blue, 0 }  ,fill opacity=1 ] (617,63) .. controls (617,60.79) and (618.79,59) .. (621,59) .. controls (623.21,59) and (625,60.79) .. (625,63) .. controls (625,65.21) and (623.21,67) .. (621,67) .. controls (618.79,67) and (617,65.21) .. (617,63) -- cycle ;
%Shape: Circle [id:dp9227747895220744] 
\draw  [fill={rgb, 255:red, 0; green, 0; blue, 0 }  ,fill opacity=1 ] (617,103) .. controls (617,100.79) and (618.79,99) .. (621,99) .. controls (623.21,99) and (625,100.79) .. (625,103) .. controls (625,105.21) and (623.21,107) .. (621,107) .. controls (618.79,107) and (617,105.21) .. (617,103) -- cycle ;
%Shape: Circle [id:dp7284844765530363] 
\draw  [fill={rgb, 255:red, 0; green, 0; blue, 0 }  ,fill opacity=1 ] (617,143) .. controls (617,140.79) and (618.79,139) .. (621,139) .. controls (623.21,139) and (625,140.79) .. (625,143) .. controls (625,145.21) and (623.21,147) .. (621,147) .. controls (618.79,147) and (617,145.21) .. (617,143) -- cycle ;
%Shape: Circle [id:dp47731842680994263] 
\draw  [fill={rgb, 255:red, 0; green, 0; blue, 0 }  ,fill opacity=1 ] (617,183) .. controls (617,180.79) and (618.79,179) .. (621,179) .. controls (623.21,179) and (625,180.79) .. (625,183) .. controls (625,185.21) and (623.21,187) .. (621,187) .. controls (618.79,187) and (617,185.21) .. (617,183) -- cycle ;
%Shape: Circle [id:dp3181991061308065] 
\draw  [fill={rgb, 255:red, 0; green, 0; blue, 0 }  ,fill opacity=1 ] (617,223) .. controls (617,220.79) and (618.79,219) .. (621,219) .. controls (623.21,219) and (625,220.79) .. (625,223) .. controls (625,225.21) and (623.21,227) .. (621,227) .. controls (618.79,227) and (617,225.21) .. (617,223) -- cycle ;

%Straight Lines [id:da9499229492894903] 
\draw    (76,102) -- (621,63) ;
%Straight Lines [id:da5489376058395936] 
\draw    (76,142) -- (621,103) ;
%Straight Lines [id:da08147650000058637] 
\draw    (76,182) -- (621,143) ;
%Straight Lines [id:da39017104904448097] 
\draw    (76,222) -- (621,183) ;
%Curve Lines [id:da22332600414350323] 
\draw  [dash pattern={on 4.5pt off 4.5pt}]  (76,62) .. controls (116,32) and (389,-13) .. (491,65) ;
%Shape: Circle [id:dp8245156453240488] 
\draw  [fill={rgb, 255:red, 0; green, 0; blue, 0 }  ,fill opacity=1 ] (490,72) .. controls (490,69.79) and (491.79,68) .. (494,68) .. controls (496.21,68) and (498,69.79) .. (498,72) .. controls (498,74.21) and (496.21,76) .. (494,76) .. controls (491.79,76) and (490,74.21) .. (490,72) -- cycle ;
%Shape: Circle [id:dp5999243798073639] 
\draw  [fill={rgb, 255:red, 0; green, 0; blue, 0 }  ,fill opacity=1 ] (445,75) .. controls (445,72.79) and (446.79,71) .. (449,71) .. controls (451.21,71) and (453,72.79) .. (453,75) .. controls (453,77.21) and (451.21,79) .. (449,79) .. controls (446.79,79) and (445,77.21) .. (445,75) -- cycle ;
%Straight Lines [id:da7930721647593626] 
\draw [line width=1.5]  [dash pattern={on 5.63pt off 4.5pt}]  (491,65) -- (447,66) ;
%Curve Lines [id:da4520534431422678] 
\draw  [dash pattern={on 4.5pt off 4.5pt}]  (341,114) .. controls (359.45,100.16) and (374.92,100.37) .. (390.16,99.3) .. controls (407.96,98.06) and (425.45,95.09) .. (447,66) ;
%Straight Lines [id:da7223793269704899] 
\draw [line width=1.5]  [dash pattern={on 5.63pt off 4.5pt}]  (341,114) -- (304,117) ;
%Shape: Circle [id:dp7235179480267186] 
\draw  [fill={rgb, 255:red, 0; green, 0; blue, 0 }  ,fill opacity=1 ] (337,123) .. controls (337,120.79) and (338.79,119) .. (341,119) .. controls (343.21,119) and (345,120.79) .. (345,123) .. controls (345,125.21) and (343.21,127) .. (341,127) .. controls (338.79,127) and (337,125.21) .. (337,123) -- cycle ;
%Shape: Circle [id:dp577757069001116] 
\draw  [fill={rgb, 255:red, 0; green, 0; blue, 0 }  ,fill opacity=1 ] (301,126) .. controls (301,123.79) and (302.79,122) .. (305,122) .. controls (307.21,122) and (309,123.79) .. (309,126) .. controls (309,128.21) and (307.21,130) .. (305,130) .. controls (302.79,130) and (301,128.21) .. (301,126) -- cycle ;
%Curve Lines [id:da2433290577416658] 
\draw  [dash pattern={on 4.5pt off 4.5pt}]  (214,164) .. controls (232.45,150.16) and (231.92,151.37) .. (247.16,150.3) .. controls (262.4,149.24) and (249,168) .. (304,117) ;
%Straight Lines [id:da1087010756576825] 
\draw [line width=1.5]  [dash pattern={on 5.63pt off 4.5pt}]  (214,164) -- (177,167) ;
%Shape: Circle [id:dp2988682282661397] 
\draw  [fill={rgb, 255:red, 0; green, 0; blue, 0 }  ,fill opacity=1 ] (210,173) .. controls (210,170.79) and (211.79,169) .. (214,169) .. controls (216.21,169) and (218,170.79) .. (218,173) .. controls (218,175.21) and (216.21,177) .. (214,177) .. controls (211.79,177) and (210,175.21) .. (210,173) -- cycle ;
%Shape: Circle [id:dp011261741491258137] 
\draw  [fill={rgb, 255:red, 0; green, 0; blue, 0 }  ,fill opacity=1 ] (174,176) .. controls (174,173.79) and (175.79,172) .. (178,172) .. controls (180.21,172) and (182,173.79) .. (182,176) .. controls (182,178.21) and (180.21,180) .. (178,180) .. controls (175.79,180) and (174,178.21) .. (174,176) -- cycle ;
%Curve Lines [id:da922904805345681] 
\draw  [dash pattern={on 4.5pt off 4.5pt}]  (483,183) .. controls (416,140) and (117,221) .. (177,167) ;
%Straight Lines [id:da7461345575932947] 
\draw [line width=1.5]  [dash pattern={on 5.63pt off 4.5pt}]  (483,183) -- (446,186) ;
%Shape: Circle [id:dp7495143323355058] 
\draw  [fill={rgb, 255:red, 0; green, 0; blue, 0 }  ,fill opacity=1 ] (479,192) .. controls (479,189.79) and (480.79,188) .. (483,188) .. controls (485.21,188) and (487,189.79) .. (487,192) .. controls (487,194.21) and (485.21,196) .. (483,196) .. controls (480.79,196) and (479,194.21) .. (479,192) -- cycle ;
%Shape: Circle [id:dp45119398245144104] 
\draw  [fill={rgb, 255:red, 0; green, 0; blue, 0 }  ,fill opacity=1 ] (443,195) .. controls (443,192.79) and (444.79,191) .. (447,191) .. controls (449.21,191) and (451,192.79) .. (451,195) .. controls (451,197.21) and (449.21,199) .. (447,199) .. controls (444.79,199) and (443,197.21) .. (443,195) -- cycle ;
%Curve Lines [id:da17578433518256142] 
\draw  [dash pattern={on 4.5pt off 4.5pt}]  (621,223) .. controls (595.57,206.68) and (525,255) .. (524,235) .. controls (523,215) and (408.78,219.5) .. (446,186) ;
%Rounded Rect [id:dp168208082316192] 
\draw   (41,327.03) .. controls (41,319.3) and (47.27,313.03) .. (55,313.03) -- (97,313.03) .. controls (104.73,313.03) and (111,319.3) .. (111,327.03) -- (111,516.03) .. controls (111,523.76) and (104.73,530.03) .. (97,530.03) -- (55,530.03) .. controls (47.27,530.03) and (41,523.76) .. (41,516.03) -- cycle ;
%Rounded Rect [id:dp3609665580410102] 
\draw   (586,328.03) .. controls (586,320.3) and (592.27,314.03) .. (600,314.03) -- (642,314.03) .. controls (649.73,314.03) and (656,320.3) .. (656,328.03) -- (656,517.03) .. controls (656,524.76) and (649.73,531.03) .. (642,531.03) -- (600,531.03) .. controls (592.27,531.03) and (586,524.76) .. (586,517.03) -- cycle ;
%Curve Lines [id:da025656380847250082] 
\draw [color={rgb, 255:red, 238; green, 17; blue, 17 }  ,draw opacity=1 ][line width=1.5]    (76,339.03) .. controls (116,309.03) and (392,267.03) .. (494,345.03) ;
%Straight Lines [id:da1852854270021982] 
\draw [color={rgb, 255:red, 184; green, 184; blue, 184 }  ,draw opacity=1 ][line width=0.75]  [dash pattern={on 4.5pt off 4.5pt}]  (491,342.03) -- (447,343.03) ;
%Curve Lines [id:da5971902738707506] 
\draw [color={rgb, 255:red, 126; green, 211; blue, 33 }  ,draw opacity=1 ][line width=1.5]    (343,400.03) .. controls (361.45,386.19) and (376.92,386.4) .. (392.16,385.33) .. controls (409.96,384.09) and (427.45,381.12) .. (449,352.03) ;
%Straight Lines [id:da10677154960752344] 
\draw [color={rgb, 255:red, 163; green, 163; blue, 163 }  ,draw opacity=1 ][line width=0.75]  [dash pattern={on 4.5pt off 4.5pt}]  (341,391.03) -- (304,394.03) ;
%Straight Lines [id:da5529610763524959] 
\draw [color={rgb, 255:red, 163; green, 163; blue, 163 }  ,draw opacity=1 ][line width=0.75]  [dash pattern={on 4.5pt off 4.5pt}]  (214,441.03) -- (177,444.03) ;
%Curve Lines [id:da8589028084301402] 
\draw [color={rgb, 255:red, 144; green, 19; blue, 254 }  ,draw opacity=1 ][line width=1.5]    (484,469.03) .. controls (417,426.03) and (118,507.03) .. (178,453.03) ;
%Straight Lines [id:da5451849644855185] 
\draw [color={rgb, 255:red, 163; green, 163; blue, 163 }  ,draw opacity=1 ][line width=0.75]  [dash pattern={on 4.5pt off 4.5pt}]  (483,460.03) -- (446,463.03) ;
%Curve Lines [id:da2682768601714973] 
\draw [color={rgb, 255:red, 174; green, 180; blue, 51 }  ,draw opacity=1 ][line width=1.5]    (621,500.03) .. controls (595.57,483.71) and (526,541.03) .. (525,521.03) .. controls (524,501.03) and (409.78,505.53) .. (447,472.03) ;
%Straight Lines [id:da39350407922699593] 
\draw [color={rgb, 255:red, 208; green, 2; blue, 27 }  ,draw opacity=1 ][line width=1.5]    (494,349.03) -- (621,340.03) ;
%Shape: Circle [id:dp3269602230061046] 
\draw  [fill={rgb, 255:red, 0; green, 0; blue, 0 }  ,fill opacity=1 ][line width=0.75]  (490,349.03) .. controls (490,346.82) and (491.79,345.03) .. (494,345.03) .. controls (496.21,345.03) and (498,346.82) .. (498,349.03) .. controls (498,351.24) and (496.21,353.03) .. (494,353.03) .. controls (491.79,353.03) and (490,351.24) .. (490,349.03) -- cycle ;
%Straight Lines [id:da15480833833204222] 
\draw [color={rgb, 255:red, 126; green, 211; blue, 33 }  ,draw opacity=1 ][line width=1.5]    (76,379.03) -- (449,352.03) ;
%Shape: Circle [id:dp6735557286534657] 
\draw  [fill={rgb, 255:red, 0; green, 0; blue, 0 }  ,fill opacity=1 ] (445,352.03) .. controls (445,349.82) and (446.79,348.03) .. (449,348.03) .. controls (451.21,348.03) and (453,349.82) .. (453,352.03) .. controls (453,354.24) and (451.21,356.03) .. (449,356.03) .. controls (446.79,356.03) and (445,354.24) .. (445,352.03) -- cycle ;
%Shape: Circle [id:dp8241230967360087] 
\draw  [fill={rgb, 255:red, 0; green, 0; blue, 0 }  ,fill opacity=1 ] (617,340.03) .. controls (617,337.82) and (618.79,336.03) .. (621,336.03) .. controls (623.21,336.03) and (625,337.82) .. (625,340.03) .. controls (625,342.24) and (623.21,344.03) .. (621,344.03) .. controls (618.79,344.03) and (617,342.24) .. (617,340.03) -- cycle ;
%Straight Lines [id:da7538504490772837] 
\draw [color={rgb, 255:red, 74; green, 144; blue, 226 }  ,draw opacity=1 ][line width=1.5]    (76,419.03) -- (305,403.03) ;
%Shape: Circle [id:dp1799984926046445] 
\draw  [fill={rgb, 255:red, 0; green, 0; blue, 0 }  ,fill opacity=1 ] (301,403.03) .. controls (301,400.82) and (302.79,399.03) .. (305,399.03) .. controls (307.21,399.03) and (309,400.82) .. (309,403.03) .. controls (309,405.24) and (307.21,407.03) .. (305,407.03) .. controls (302.79,407.03) and (301,405.24) .. (301,403.03) -- cycle ;
%Shape: Circle [id:dp9952550518744252] 
\draw  [fill={rgb, 255:red, 0; green, 0; blue, 0 }  ,fill opacity=1 ][line width=0.75]  (337,400.03) .. controls (337,397.82) and (338.79,396.03) .. (341,396.03) .. controls (343.21,396.03) and (345,397.82) .. (345,400.03) .. controls (345,402.24) and (343.21,404.03) .. (341,404.03) .. controls (338.79,404.03) and (337,402.24) .. (337,400.03) -- cycle ;
%Shape: Circle [id:dp8577230447123825] 
\draw  [fill={rgb, 255:red, 0; green, 0; blue, 0 }  ,fill opacity=1 ] (617,420.03) .. controls (617,417.82) and (618.79,416.03) .. (621,416.03) .. controls (623.21,416.03) and (625,417.82) .. (625,420.03) .. controls (625,422.24) and (623.21,424.03) .. (621,424.03) .. controls (618.79,424.03) and (617,422.24) .. (617,420.03) -- cycle ;
%Shape: Circle [id:dp4233704670108436] 
\draw  [fill={rgb, 255:red, 0; green, 0; blue, 0 }  ,fill opacity=1 ] (617,380.03) .. controls (617,377.82) and (618.79,376.03) .. (621,376.03) .. controls (623.21,376.03) and (625,377.82) .. (625,380.03) .. controls (625,382.24) and (623.21,384.03) .. (621,384.03) .. controls (618.79,384.03) and (617,382.24) .. (617,380.03) -- cycle ;
%Straight Lines [id:da48900580680728045] 
\draw [color={rgb, 255:red, 144; green, 19; blue, 254 }  ,draw opacity=1 ][line width=1.5]    (76,459.03) -- (178,453.03) ;
%Shape: Circle [id:dp6413299656650304] 
\draw  [fill={rgb, 255:red, 0; green, 0; blue, 0 }  ,fill opacity=1 ] (174,453.03) .. controls (174,450.82) and (175.79,449.03) .. (178,449.03) .. controls (180.21,449.03) and (182,450.82) .. (182,453.03) .. controls (182,455.24) and (180.21,457.03) .. (178,457.03) .. controls (175.79,457.03) and (174,455.24) .. (174,453.03) -- cycle ;
%Straight Lines [id:da14275000578701746] 
\draw [color={rgb, 255:red, 144; green, 19; blue, 254 }  ,draw opacity=1 ][line width=1.5]    (484,469.03) -- (621,460.03) ;
%Shape: Circle [id:dp16348939839082943] 
\draw  [fill={rgb, 255:red, 0; green, 0; blue, 0 }  ,fill opacity=1 ][line width=0.75]  (479,469.03) .. controls (479,466.82) and (480.79,465.03) .. (483,465.03) .. controls (485.21,465.03) and (487,466.82) .. (487,469.03) .. controls (487,471.24) and (485.21,473.03) .. (483,473.03) .. controls (480.79,473.03) and (479,471.24) .. (479,469.03) -- cycle ;
%Shape: Circle [id:dp6376233197710157] 
\draw  [fill={rgb, 255:red, 0; green, 0; blue, 0 }  ,fill opacity=1 ] (617,460.03) .. controls (617,457.82) and (618.79,456.03) .. (621,456.03) .. controls (623.21,456.03) and (625,457.82) .. (625,460.03) .. controls (625,462.24) and (623.21,464.03) .. (621,464.03) .. controls (618.79,464.03) and (617,462.24) .. (617,460.03) -- cycle ;
%Straight Lines [id:da6641679303927506] 
\draw [color={rgb, 255:red, 174; green, 180; blue, 51 }  ,draw opacity=1 ][line width=1.5]    (76,499.03) -- (447,472.03) ;
%Shape: Circle [id:dp028034883271796707] 
\draw  [fill={rgb, 255:red, 0; green, 0; blue, 0 }  ,fill opacity=1 ] (443,472.03) .. controls (443,469.82) and (444.79,468.03) .. (447,468.03) .. controls (449.21,468.03) and (451,469.82) .. (451,472.03) .. controls (451,474.24) and (449.21,476.03) .. (447,476.03) .. controls (444.79,476.03) and (443,474.24) .. (443,472.03) -- cycle ;
%Shape: Circle [id:dp9274018387890715] 
\draw  [fill={rgb, 255:red, 245; green, 166; blue, 35 }  ,fill opacity=1 ][line width=1.5]  (72,339.03) .. controls (72,336.82) and (73.79,335.03) .. (76,335.03) .. controls (78.21,335.03) and (80,336.82) .. (80,339.03) .. controls (80,341.24) and (78.21,343.03) .. (76,343.03) .. controls (73.79,343.03) and (72,341.24) .. (72,339.03) -- cycle ;
%Shape: Circle [id:dp3279937010131384] 
\draw  [fill={rgb, 255:red, 0; green, 0; blue, 0 }  ,fill opacity=1 ] (72,379.03) .. controls (72,376.82) and (73.79,375.03) .. (76,375.03) .. controls (78.21,375.03) and (80,376.82) .. (80,379.03) .. controls (80,381.24) and (78.21,383.03) .. (76,383.03) .. controls (73.79,383.03) and (72,381.24) .. (72,379.03) -- cycle ;
%Shape: Circle [id:dp8839097455400393] 
\draw  [fill={rgb, 255:red, 0; green, 0; blue, 0 }  ,fill opacity=1 ] (72,419.03) .. controls (72,416.82) and (73.79,415.03) .. (76,415.03) .. controls (78.21,415.03) and (80,416.82) .. (80,419.03) .. controls (80,421.24) and (78.21,423.03) .. (76,423.03) .. controls (73.79,423.03) and (72,421.24) .. (72,419.03) -- cycle ;
%Shape: Circle [id:dp06787581779754381] 
\draw  [fill={rgb, 255:red, 0; green, 0; blue, 0 }  ,fill opacity=1 ] (72,459.03) .. controls (72,456.82) and (73.79,455.03) .. (76,455.03) .. controls (78.21,455.03) and (80,456.82) .. (80,459.03) .. controls (80,461.24) and (78.21,463.03) .. (76,463.03) .. controls (73.79,463.03) and (72,461.24) .. (72,459.03) -- cycle ;
%Shape: Circle [id:dp7614050143396158] 
\draw  [fill={rgb, 255:red, 0; green, 0; blue, 0 }  ,fill opacity=1 ] (72,499.03) .. controls (72,496.82) and (73.79,495.03) .. (76,495.03) .. controls (78.21,495.03) and (80,496.82) .. (80,499.03) .. controls (80,501.24) and (78.21,503.03) .. (76,503.03) .. controls (73.79,503.03) and (72,501.24) .. (72,499.03) -- cycle ;
%Shape: Circle [id:dp5735198159145034] 
\draw  [fill={rgb, 255:red, 0; green, 0; blue, 0 }  ,fill opacity=1 ] (617,500.03) .. controls (617,497.82) and (618.79,496.03) .. (621,496.03) .. controls (623.21,496.03) and (625,497.82) .. (625,500.03) .. controls (625,502.24) and (623.21,504.03) .. (621,504.03) .. controls (618.79,504.03) and (617,502.24) .. (617,500.03) -- cycle ;
%Shape: Circle [id:dp5297258035182131] 
\draw  [fill={rgb, 255:red, 0; green, 0; blue, 0 }  ,fill opacity=1 ][line width=0.75]  (210,450.03) .. controls (210,447.82) and (211.79,446.03) .. (214,446.03) .. controls (216.21,446.03) and (218,447.82) .. (218,450.03) .. controls (218,452.24) and (216.21,454.03) .. (214,454.03) .. controls (211.79,454.03) and (210,452.24) .. (210,450.03) -- cycle ;
%Shape: Square [id:dp2674499490639677] 
\draw   (8,32) -- (33.5,32) -- (33.5,57.5) -- (8,57.5) -- cycle ;
%Shape: Square [id:dp16120486712861415] 
\draw   (8,315) -- (33.5,315) -- (33.5,340.5) -- (8,340.5) -- cycle ;
%Straight Lines [id:da2929347909632427] 
\draw [color={rgb, 255:red, 126; green, 211; blue, 33 }  ,draw opacity=1 ][line width=1.5]    (343,400.03) -- (621,380.03) ;
%Straight Lines [id:da20145915070164766] 
\draw [color={rgb, 255:red, 74; green, 144; blue, 226 }  ,draw opacity=1 ][line width=1.5]    (215,450.03) -- (621,420.03) ;
%Curve Lines [id:da08948013475522387] 
\draw [color={rgb, 255:red, 74; green, 144; blue, 226 }  ,draw opacity=1 ][line width=1.5]    (215,450.03) .. controls (233.45,436.19) and (232.92,437.4) .. (248.16,436.33) .. controls (263.4,435.27) and (250,454.03) .. (305,403.03) ;
%Shape: Circle [id:dp2370445870864164] 
\draw  [fill={rgb, 255:red, 245; green, 166; blue, 35 }  ,fill opacity=1 ][line width=1.5]  (403,354) .. controls (403,351.79) and (404.79,350) .. (407,350) .. controls (409.21,350) and (411,351.79) .. (411,354) .. controls (411,356.21) and (409.21,358) .. (407,358) .. controls (404.79,358) and (403,356.21) .. (403,354) -- cycle ;
%Shape: Circle [id:dp9283818365838143] 
\draw  [fill={rgb, 255:red, 245; green, 166; blue, 35 }  ,fill opacity=1 ][line width=1.5]  (262,405) .. controls (262,402.79) and (263.79,401) .. (266,401) .. controls (268.21,401) and (270,402.79) .. (270,405) .. controls (270,407.21) and (268.21,409) .. (266,409) .. controls (263.79,409) and (262,407.21) .. (262,405) -- cycle ;
%Shape: Circle [id:dp7456961199077536] 
\draw  [fill={rgb, 255:red, 245; green, 166; blue, 35 }  ,fill opacity=1 ][line width=1.5]  (139,455) .. controls (139,452.79) and (140.79,451) .. (143,451) .. controls (145.21,451) and (147,452.79) .. (147,455) .. controls (147,457.21) and (145.21,459) .. (143,459) .. controls (140.79,459) and (139,457.21) .. (139,455) -- cycle ;
%Shape: Circle [id:dp16854482536012383] 
\draw  [fill={rgb, 255:red, 245; green, 166; blue, 35 }  ,fill opacity=1 ][line width=1.5]  (405,475) .. controls (405,472.79) and (406.79,471) .. (409,471) .. controls (411.21,471) and (413,472.79) .. (413,475) .. controls (413,477.21) and (411.21,479) .. (409,479) .. controls (406.79,479) and (405,477.21) .. (405,475) -- cycle ;
%Straight Lines [id:da22389595216655667] 
\draw    (368,494) -- (346.26,520.94) ;
\draw [shift={(345,522.5)}, rotate = 308.9] [color={rgb, 255:red, 0; green, 0; blue, 0 }  ][line width=0.75]    (10.93,-3.29) .. controls (6.95,-1.4) and (3.31,-0.3) .. (0,0) .. controls (3.31,0.3) and (6.95,1.4) .. (10.93,3.29)   ;

% Text Node
\draw (19,132.9) node [anchor=north west][inner sep=0.75pt]    {$X$};
% Text Node
\draw (665,133.9) node [anchor=north west][inner sep=0.75pt]    {$Y$};
% Text Node
\draw (16,409.9) node [anchor=north west][inner sep=0.75pt]    {$X$};
% Text Node
\draw (662,410.9) node [anchor=north west][inner sep=0.75pt]    {$Y$};
% Text Node
\draw (112.2,345) node [anchor=north west][inner sep=0.75pt]  [color={rgb, 255:red, 126; green, 211; blue, 33 }  ,opacity=1 ,rotate=-354.31]  {$P'_{1}$};
% Text Node
\draw (113.83,386) node [anchor=north west][inner sep=0.75pt]  [color={rgb, 255:red, 74; green, 144; blue, 226 }  ,opacity=1 ,rotate=-354.31]  {$P'_{2}$};
% Text Node
\draw (113.83,425) node [anchor=north west][inner sep=0.75pt]  [color={rgb, 255:red, 144; green, 19; blue, 254 }  ,opacity=1 ,rotate=-354.31]  {$P'_{3}$};
% Text Node
\draw (113.83,500) node [anchor=north west][inner sep=0.75pt]  [color={rgb, 255:red, 174; green, 180; blue, 51 }  ,opacity=1 ,rotate=-354.31]  {$P'_{4}$};
% Text Node
\draw (230.89,271) node [anchor=north west][inner sep=0.75pt]  [color={rgb, 255:red, 208; green, 2; blue, 27 }  ,opacity=1 ,rotate=-350.71]  {$P'_{0}$};
% Text Node
\draw (34,45) node  [font=\scriptsize] [align=left] {\begin{minipage}[lt]{15.98pt}\setlength\topsep{0pt}
a:
\end{minipage}};
% Text Node
\draw (34,325) node  [font=\scriptsize] [align=left] {\begin{minipage}[lt]{15.64pt}\setlength\topsep{0pt}
b:
\end{minipage}};
% Text Node
\draw (393.22,328.84) node [anchor=north west][inner sep=0.75pt]  [color={rgb, 255:red, 245; green, 166; blue, 35 }  ,opacity=1 ,rotate=-354.98]  {$c_{P_{1}}$};
% Text Node
\draw (252.29,380.21) node [anchor=north west][inner sep=0.75pt]  [color={rgb, 255:red, 245; green, 166; blue, 35 }  ,opacity=1 ,rotate=-354.98]  {$c_{P_{2}}$};
% Text Node
\draw (131.22,461.84) node [anchor=north west][inner sep=0.75pt]  [color={rgb, 255:red, 245; green, 166; blue, 35 }  ,opacity=1 ,rotate=-354.98]  {$c_{P_{3}}$};
% Text Node
\draw (396.22,482.84) node [anchor=north west][inner sep=0.75pt]  [color={rgb, 255:red, 245; green, 166; blue, 35 }  ,opacity=1 ,rotate=-354.98]  {$c_{P_{4}}$};
% Text Node
\draw (114.2,74.43) node [anchor=north west][inner sep=0.75pt]  [rotate=-354.31]  {$P_{1}$};
% Text Node
\draw (113.83,113.43) node [anchor=north west][inner sep=0.75pt]  [rotate=-354.31]  {$P_{2}$};
% Text Node
\draw (113.83,154) node [anchor=north west][inner sep=0.75pt]  [rotate=-354.31]  {$P_{3}$};
% Text Node
\draw (113.83,193.43) node [anchor=north west][inner sep=0.75pt]  [rotate=-354.31]  {$P_{4}$};
% Text Node
\draw (214.38,0) node [anchor=north west][inner sep=0.75pt]  [rotate=-355.09]  {$Q$};
% Text Node
\draw (54,311) node [anchor=north west][inner sep=0.75pt]  [color={rgb, 255:red, 245; green, 166; blue, 35 }  ,opacity=1 ]  {$c_{Q}$};
% Text Node
\draw (183,525) node [anchor=north west][inner sep=0.75pt]  [font=\small,color={rgb, 255:red, 0; green, 0; blue, 0 }  ,opacity=1 ] [align=left] {Separator $\displaystyle C_{\mathcal{P}} \ =\ \{c_{Q}$$\displaystyle ,c_{P}{}_{1} ,c_{P}{}_{2} ,c_{P}{}_{3} ,c_{P}{}_{4}\}$. \\Leftmost minimum size separator. };

\end{tikzpicture}
    \caption{Augmenting}
    \label{fig:augmenting}
\end{figure}
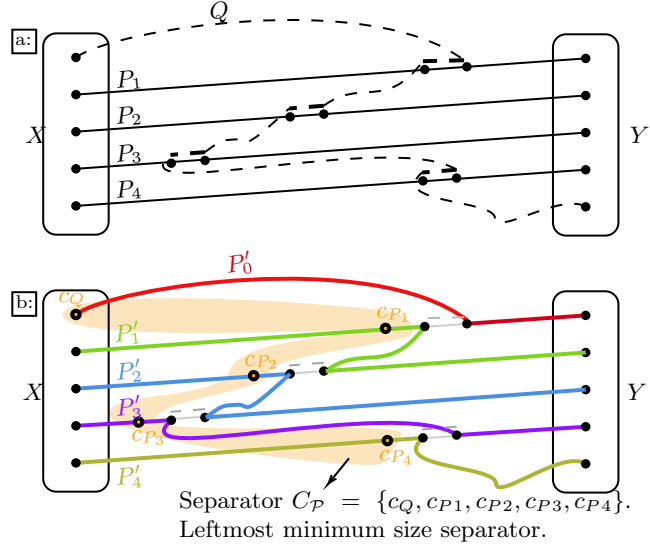

Now that we have the maximum size family of pairwise disjoint paths from $X$ to $Y$ (namely $\mathcal{P}$), we show how to find the leftmost minimum sized ($X$, $Y$, $\leq k$)$^G$-separator\footnote{Again, for more details refer to \cite{FlumGrohe}.}.

Let $R(\mathcal{P})$\footnote{We have borrowed the notation for this part from \cite{FlumGrohe}.} be the set of all vertices $v \in V$ such that there is a $\mathcal{P}$-augmenting path from $X$ to $v$.
%%For every disjoint path
For every path $P\in \mathcal{P}$, let $c_P$ be the first vertex on $P$ not contained in $R(\mathcal{P})$. Also, let $C(\mathcal{P}) = \bigcup\limits_{P \in \mathcal{P}} \{c_P\}$.

Note that if $c_P \notin X$, then $c_P$ lies on an augmenting path that continues to the predecessor of $c_P$ on $P$.
\begin{claim}(Claim 2 of Lemma 11.20 in \cite{FlumGrohe})
	$C(\mathcal{P})$ separates $X$ from $Y$.
\end{claim}

So, $C(\mathcal{P})$ is an ($X$, $Y$, $\leq k$)-separator and also based on Menger's result, this is a separator of the minimum size. In fact, $C(\mathcal{P})$ is the leftmost minimum size separator due to the construction of the $c_P$'s (check $C_{\mathcal{P}}$ in Figure~\ref{fig:augmenting} and $S$ in Figure~\ref{fig:all-seps}).

Now, we present pseudocode for the algorithm described above with a tiny modification. The modification is that we do not start from scratch necessarily and we pass a subset $\mathcal{P}$ of disjoint paths to the algorithm in the beginning. Then, the algorithm tries to find a $\mathcal{P}$-augmenting walk $Q$,  updates $\mathcal{P}$, and recurses until no augmenting walk is found. The original algorithm can be implemented by initializing $\mathcal{P} = \emptyset$. The reason why we need this modification is described in Section \ref{sec:min_left}.

The following are the global variables used in Algorithms \ref{algo:1}, \ref{algo:2}, and \ref{algo:3}.
\begin{itemize}
	\item $X\subseteq V(G)$: the left set
	\item $\mathcal{A}\subseteq 2^{V(G)}$: the set of all minimal leftmost ($X$, $Y$, $\leq k$)$^G$-separators 
	\item $k \in \mathbb{N}$: a given upper bound on the treewidth
\end{itemize}

\textbf{Observation.} Note that the output of Algorithm \ref{algo:1} is unique and well-defined.

\begin{figure}
    \centering
    \tikzset{every picture/.style={line width=0.75pt}} %set default line width to 0.75pt        

\begin{tikzpicture}[x=0.5pt,y=0.5pt,yscale=-1,xscale=1]
%uncomment if require: \path (0,434); %set diagram left start at 0, and has height of 434

%Shape: Polygon Curved [id:ds534884010254741] 
\draw  [fill={rgb, 255:red, 86; green, 14; blue, 235 }  ,fill opacity=0.2 ] (173,64) .. controls (193,54) and (250,95) .. (230,115) .. controls (210,135) and (268,182) .. (288,212) .. controls (308,242) and (211,262) .. (224,231) .. controls (237,200) and (153,74) .. (173,64) -- cycle ;
%Shape: Polygon Curved [id:ds33247241624713153] 
\draw  [fill={rgb, 255:red, 246; green, 137; blue, 137 }  ,fill opacity=0.2 ] (172,139) .. controls (192,129) and (324,66) .. (304,86) .. controls (284,106) and (192,156) .. (212,186) .. controls (232,216) and (172,282) .. (152,252) .. controls (132,222) and (152,149) .. (172,139) -- cycle ;
%Shape: Ellipse [id:dp9413658580037667] 
\draw  [fill={rgb, 255:red, 155; green, 155; blue, 155 }  ,fill opacity=0.2 ] (30,157.5) .. controls (30,90.4) and (45.67,36) .. (65,36) .. controls (84.33,36) and (100,90.4) .. (100,157.5) .. controls (100,224.6) and (84.33,279) .. (65,279) .. controls (45.67,279) and (30,224.6) .. (30,157.5) -- cycle ;
%Shape: Ellipse [id:dp51191507349334] 
\draw  [fill={rgb, 255:red, 128; green, 128; blue, 128 }  ,fill opacity=0.2 ] (596,142.5) .. controls (596,121.79) and (606.97,105) .. (620.5,105) .. controls (634.03,105) and (645,121.79) .. (645,142.5) .. controls (645,163.21) and (634.03,180) .. (620.5,180) .. controls (606.97,180) and (596,163.21) .. (596,142.5) -- cycle ;
%Curve Lines [id:da7235971718402576] 
\draw  [dash pattern={on 4.5pt off 4.5pt}]  (65,36) .. controls (148,6) and (143,61) .. (184,63) ;
%Curve Lines [id:da18205035935789282] 
\draw  [dash pattern={on 4.5pt off 4.5pt}]  (184,63) .. controls (298,72) and (355,95) .. (391,120) ;
%Curve Lines [id:da21441021342077593] 
\draw  [dash pattern={on 4.5pt off 4.5pt}]  (391,120) .. controls (449,152.5) and (471,62.5) .. (499,69.5) ;
%Curve Lines [id:da14300816635047786] 
\draw  [dash pattern={on 4.5pt off 4.5pt}]  (65,279) .. controls (103,285) and (124,261) .. (165,259) ;
%Curve Lines [id:da6460543819717817] 
\draw  [dash pattern={on 4.5pt off 4.5pt}]  (165,259) .. controls (203,265) and (211,244) .. (245,247) ;
%Curve Lines [id:da8588149735430932] 
\draw  [dash pattern={on 4.5pt off 4.5pt}]  (245,247) .. controls (306,245) and (386,155) .. (423,174) ;
%Curve Lines [id:da6780228210888601] 
\draw  [dash pattern={on 4.5pt off 4.5pt}]  (423,174) .. controls (499,205.5) and (526.5,165.5) .. (555,166.5) ;
%Shape: Ellipse [id:dp9931919223581984] 
\draw  [fill={rgb, 255:red, 3; green, 252; blue, 200 }  ,fill opacity=0.3 ] (399,149.5) .. controls (399,137.63) and (405.27,128) .. (413,128) .. controls (420.73,128) and (427,137.63) .. (427,149.5) .. controls (427,161.37) and (420.73,171) .. (413,171) .. controls (405.27,171) and (399,161.37) .. (399,149.5) -- cycle ;
%Curve Lines [id:da31256928199011136] 
\draw  [dash pattern={on 4.5pt off 4.5pt}]  (499,69.5) .. controls (541,76.5) and (492,121.5) .. (554,124.5) ;
%Shape: Ellipse [id:dp4629849029439401] 
\draw  [fill={rgb, 255:red, 232; green, 252; blue, 3 }  ,fill opacity=0.3 ][dash pattern={on 0.84pt off 2.51pt}] (534,145.5) .. controls (534,133.63) and (540.27,124) .. (548,124) .. controls (555.73,124) and (562,133.63) .. (562,145.5) .. controls (562,157.37) and (555.73,167) .. (548,167) .. controls (540.27,167) and (534,157.37) .. (534,145.5) -- cycle ;
%Curve Lines [id:da12662197437128775] 
\draw  [dash pattern={on 4.5pt off 4.5pt}]  (554,124.5) .. controls (601,117.5) and (570,107) .. (620.5,105) ;
%Curve Lines [id:da47873704740752276] 
\draw  [dash pattern={on 4.5pt off 4.5pt}]  (555,166.5) .. controls (582,166.5) and (604,180.5) .. (620.5,180) ;
%Straight Lines [id:da6894170543468678] 
\draw    (306,82) -- (375,78.5) ;
%Shape: Ellipse [id:dp7363953284040996] 
\draw  [fill={rgb, 255:red, 255; green, 8; blue, 42 }  ,fill opacity=0.33 ] (375,78.5) .. controls (370.94,68.23) and (382.21,54.14) .. (400.19,47.02) .. controls (418.16,39.91) and (436.02,42.47) .. (440.09,52.74) .. controls (444.15,63.01) and (432.88,77.11) .. (414.9,84.22) .. controls (396.93,91.33) and (379.06,88.77) .. (375,78.5) -- cycle ;
%Straight Lines [id:da8233657503488285] 
\draw    (245,193) -- (312.83,287.38) ;
\draw [shift={(314,289)}, rotate = 234.29] [color={rgb, 255:red, 0; green, 0; blue, 0 }  ][line width=0.75]    (10.93,-3.29) .. controls (6.95,-1.4) and (3.31,-0.3) .. (0,0) .. controls (3.31,0.3) and (6.95,1.4) .. (10.93,3.29)   ;
%Straight Lines [id:da6636160731295031] 
\draw    (196,222) -- (307.31,291.94) ;
\draw [shift={(309,293)}, rotate = 212.14] [color={rgb, 255:red, 0; green, 0; blue, 0 }  ][line width=0.75]    (10.93,-3.29) .. controls (6.95,-1.4) and (3.31,-0.3) .. (0,0) .. controls (3.31,0.3) and (6.95,1.4) .. (10.93,3.29)   ;
%Straight Lines [id:da3958716982256285] 
\draw    (427,55) -- (493.06,38.49) ;
\draw [shift={(495,38)}, rotate = 525.96] [color={rgb, 255:red, 0; green, 0; blue, 0 }  ][line width=0.75]    (10.93,-3.29) .. controls (6.95,-1.4) and (3.31,-0.3) .. (0,0) .. controls (3.31,0.3) and (6.95,1.4) .. (10.93,3.29)   ;
%Straight Lines [id:da16669286122291926] 
\draw    (413,160) -- (415.82,191.51) ;
\draw [shift={(416,193.5)}, rotate = 264.88] [color={rgb, 255:red, 0; green, 0; blue, 0 }  ][line width=0.75]    (10.93,-3.29) .. controls (6.95,-1.4) and (3.31,-0.3) .. (0,0) .. controls (3.31,0.3) and (6.95,1.4) .. (10.93,3.29)   ;
%Straight Lines [id:da3038484839935185] 
\draw    (549,155.5) -- (574.8,189.9) ;
\draw [shift={(576,191.5)}, rotate = 233.13] [color={rgb, 255:red, 0; green, 0; blue, 0 }  ][line width=0.75]    (10.93,-3.29) .. controls (6.95,-1.4) and (3.31,-0.3) .. (0,0) .. controls (3.31,0.3) and (6.95,1.4) .. (10.93,3.29)   ;
%Rounded Rect [id:dp06526310696429083] 
\draw  [fill={rgb, 255:red, 254; green, 248; blue, 4 }  ,fill opacity=0.06 ] (278,12.8) .. controls (278,9.6) and (280.6,7) .. (283.8,7) -- (378.2,7) .. controls (381.4,7) and (384,9.6) .. (384,12.8) -- (384,30.2) .. controls (384,33.4) and (381.4,36) .. (378.2,36) -- (283.8,36) .. controls (280.6,36) and (278,33.4) .. (278,30.2) -- cycle ;
%Straight Lines [id:da8176637531119368] 
\draw    (290,22) -- (226,21.03) ;
\draw [shift={(224,21)}, rotate = 360.87] [color={rgb, 255:red, 0; green, 0; blue, 0 }  ][line width=0.75]    (10.93,-3.29) .. controls (6.95,-1.4) and (3.31,-0.3) .. (0,0) .. controls (3.31,0.3) and (6.95,1.4) .. (10.93,3.29)   ;
%Rounded Rect [id:dp008419611110936254] 
\draw  [fill={rgb, 255:red, 254; green, 248; blue, 4 }  ,fill opacity=0.06 ] (262,316.4) .. controls (262,312.31) and (265.31,309) .. (269.4,309) -- (417.6,309) .. controls (421.69,309) and (425,312.31) .. (425,316.4) -- (425,338.6) .. controls (425,342.69) and (421.69,346) .. (417.6,346) -- (269.4,346) .. controls (265.31,346) and (262,342.69) .. (262,338.6) -- cycle ;

% Text Node
\draw (57,141.4) node [anchor=north west][inner sep=0.75pt]    {$X$};
% Text Node
\draw (540,136.4) node [anchor=north west][inner sep=0.75pt]    {$S'$};
% Text Node
\draw (408,139.4) node [anchor=north west][inner sep=0.75pt]    {$S$};
% Text Node
\draw (188,74.4) node [anchor=north west][inner sep=0.75pt]    {$S_{1}$};
% Text Node
\draw (165,222.4) node [anchor=north west][inner sep=0.75pt]    {$S_{2}$};
% Text Node
\draw (615,133.4) node [anchor=north west][inner sep=0.75pt]    {$Y$};
% Text Node
\draw (398,53.4) node [anchor=north west][inner sep=0.75pt]    {$G_{1}$};
% Text Node
\draw (312,285.5) node [anchor=north west][inner sep=0.75pt]   [align=left] {{\scriptsize Two different leftmost separators [non-comparable]}};
% Text Node
\draw (502,6) node [anchor=north west][inner sep=0.75pt]   [align=left] {{\scriptsize The subgraph that gets}\\{\scriptsize disconnected from $\displaystyle X\ $}\\{\scriptsize and $\displaystyle Y\ $after removing $\displaystyle S_{2}$}};
% Text Node
\draw (365,193) node [anchor=north west][inner sep=0.75pt]   [align=left] {{\scriptsize The leftmost minimum}\\{\scriptsize size separator (output}\\{\scriptsize  of Algorithm 1)}\\{\scriptsize  [This is unique]}};
% Text Node
\draw (548,194) node [anchor=north west][inner sep=0.75pt]   [align=left] {{\scriptsize A minimum size}\\{\scriptsize separator}};
% Text Node
\draw (280.5,13.4) node [anchor=north west][inner sep=0.75pt]  [font=\scriptsize]  {$w(X)>w(Y)$};
% Text Node
\draw (139,13) node [anchor=north west][inner sep=0.75pt]   [align=left] {{\scriptsize weight of $\displaystyle X$}};
% Text Node
\draw (272,320.4) node [anchor=north west][inner sep=0.75pt]  [font=\scriptsize]  {$| S_{1}| ,\ |S_{2} |,\ |S|,\ |S'|\ \leq \ k$};

\end{tikzpicture}
    \caption{Leftmost separators vs leftmost minimum size separator}
    \label{fig:all-seps}
\end{figure}
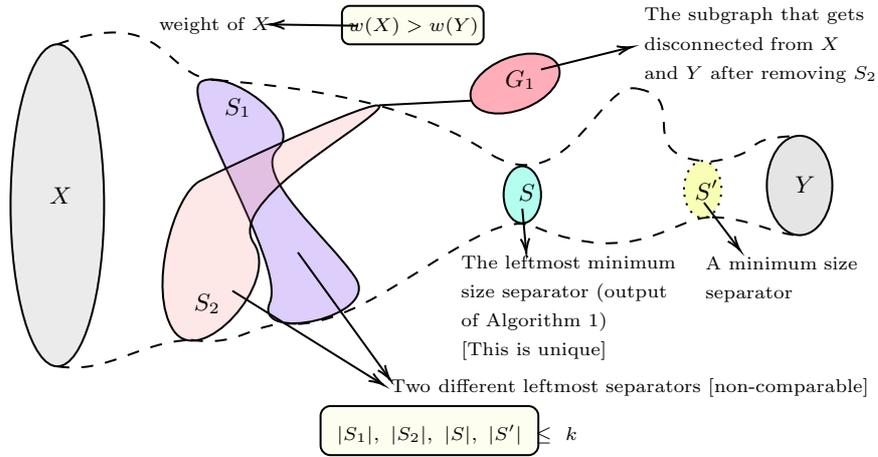
\begin{algorithm}
		\DontPrintSemicolon % Some LaTeX compilers require you to use \dontprintsemicolon instead
		\KwIn{Graph $G = (V, E)$, a subset of vertices $Y$, and a set $\mathcal{P}$ of pairwise disjoint paths from $X$ to $Y$}
		\KwOut{The leftmost minimum ($X$, $Y$, $\leq k$)$^G$-separator, and an updated set of pairwise disjoint paths $\mathcal{P}$}
		\While{$\exists$ a $\mathcal{P}$-augmenting walk $Q$}{Update $\mathcal{P}$.}\tcp*{by sending a unit flow through the edges of $\mathcal{P}$ and $Q$. Also, $|\mathcal{P}|$ is increased by 1}
		Construct $R(\mathcal{P}) = \{v \in V\ | \ \exists \text{ a } \mathcal{P}\text{-augmenting walk from } X \text{ to } v\}$ using DFS.\;
		Initialize $C(\mathcal{P})\leftarrow \emptyset$\;
		\For{$P \in \mathcal{P}$}{$c_P \leftarrow $ the first vertex $\in P$ and $\notin R(\mathcal{P})$.\;
			$C(\mathcal{P})\leftarrow C(\mathcal{P}) \cup \{c_P\}$.}
		\Return{$(C(\mathcal{P}), \mathcal{P})$.}
		\caption{{\sc $Left\_Minimum\_Sep(G, Y, \mathcal{P})$: Find the leftmost minimum separator} }
		\label{algo:1}
	\end{algorithm}
\begin{algorithm}
	\DontPrintSemicolon % Some LaTeX compilers require you to use \dontprintsemicolon instead
	\KwIn{Graph $G$}
	%\KwOut{flag $k$-separators (in $S_{X, Y}^{G}$)}
	$\{S, \mathcal{P}\} \gets Left\_Minimum\_Sep(G, Y, \emptyset)$\; \tcp*{$S$ is the minimum size leftmost ($X$, $Y$, $\leq k$)-separator and $\mathcal{P}$ is a set of pairwise disjoint paths}
	push all the vertices in $S$ onto  empty stack $R$.\;
	$Branch$($G[V_{X, S}\cup S], S, \emptyset, \mathcal{P}, $R$, True$)\;
	\caption{{\sc $Init(G, Y)$: Initialization}}
	\label{algo:2}
\end{algorithm}

\begin{algorithm}
	\DontPrintSemicolon % Some LaTeX compilers require you to use \dontprintsemicolon instead
	\KwIn{Graph $G$, a separator $S \in S_{X,Y}$, a subset of vertices $Y$, $I$: the included vertices, $\mathcal{P}$: a set of pairwise disjoint paths between $X$ and $Y$, $R$: the stack to hold the order of handling vertices, and $\mathit{leftmost}$: a boolean indicating whether we have a leftmost separator.}
	%\KwOut{flag $k$-separators (in $S_{X, Y}^{G}$)}
	\If{$I == S$}{
		$A \leftarrow A \cup \{S\}$}
	\Else{
		pop $v$ from $R$\;
		$Y' \leftarrow \left(S\setminus \{v\}\right) \cup \left(N\left(v\right) \cap V_{X, S}\right)$\;
		$\{S', \mathcal{P'}\} \gets \textit{Left\_Minimum\_Sep\/}(G, Y', \mathcal{P})$\; 
		%\If{$|S| > k$}{\Return{$False$}}
		%\If{$I \not\subset S$}{\Return{$False$}}
		\If{$|S'| \leq k \, \wedge I \subseteq S'$}{
			$\textit{leftmost}\leftarrow \textit{False}$\;
			\If{$|S'| < k$}{
				let $R'$ be a copy of $R$. Push all vertices of $\left(N\left(v\right) \cap V_{X, S}\right)$ onto $R'$\;
				$Branch\left(G[V_{X, S'}], S', Y', I, \mathcal{P'}, R', True\right)$}}
		%$Branch\left(G[V_{X, S}], S, \left(S\setminus \{v\}\right) \cup N\left(v\right), I, P\right)$\;
		\If{$(|S\setminus I| \geq 2 \, \vee \mathit{leftmost}) \wedge (|S| \leq k)$}{
			$Branch\left(G[V_{X, S}], S, Y, I\cup \{v\}, \mathcal{P}, R, \mathit{leftmost}\right)$
			%$Branch\left(G[V_{X, S}], S, I\cup \{v\}, P, flag\right)$
	}}
	
	\caption{{\sc $Branch(G, S, Y, I, \mathcal{P}, R,\textit{leftmost})$: the Main Procedure for finding all leftmost ($X$, $Y$, $\leq k$)-separators}}
	\label{algo:3}
\end{algorithm}

\section{Finding All Minimal Leftmost ($X$, $Y$, $\leq k$)-Separators}\label{sec:min_left}

In this section, we present our main algorithm. In the introduction, we mentioned why it is important to find the leftmost\footnote{We drop the term ``minimal'' because it is the default for the separators throughout this paper unless mentioned otherwise.} separators. Also, in Section \ref{sec:flumgrohe} we reviewed an algorithm (Algorithm \ref{algo:1}) to find the leftmost minimum separator. We use this algorithm in ours. 

In our problem, we have two subsets of vertices $X$ and $Y$ such that $|X| \geq |Y|$. W.l.o.g., assume $X$ is the set on the left and $Y$ is the set on the right. 

\begin{tcolorbox}
	\textbf{Problem 2.} Given a graph $G = (V, E)$, sets $X, Y \subseteq V$, and $k \in \mathbb{N}$, what are the minimal leftmost ($X$, $Y$, $\leq k$)$^G$-separators?
\end{tcolorbox}
\begin{theorem}
	Given a graph $G = (V, E)$, sets $X, Y \subseteq V$, and $k \in \mathbb{N}$, there exists an algorithm which solves Problem 2 in time $\mathcal{O}(2^{2k}\sqrt{k}n)$
\end{theorem}
\begin{proof}
	We present a recursive branching algorithm (Algorithm \ref{algo:3}). Initially, it calls Algorithm \ref{algo:1} to find the leftmost minimum size ($X$, $Y$, $\leq k$)-separator (using a simple flow algorithm), namely $S$ by feeding $X$ as the left set, $Y$ as the right set, and an empty set of pairwise disjoint paths (namely $\mathcal{P}$) from $X$ to $Y$ (inherited from the parent branch). Notice that the leftmost minimum size separator is unique. This is the root (namely $r$) of the computation tree (namely $\mathcal{T}$). Let us refer to the computation subtree rooted at node $x$, as $\mathcal{T}(x)$. 
	
	In each node $x$ of $\mathcal{T}$ with the corresponding graph $G_{x}$, let $Y_{x}$, $I_{x}$, and $\mathcal{P}_x$ be the right set, the set of vertices that we require to be in the separator, and the set of disjoint paths inherited from the parent's node, respectively. Notice that we do not pass $X_x$ (the left set) as an argument since it does not change throughout the algorithm and hence we have defined it as a global variable ($\forall x, X_x = X$). 
	
	\begin{claim}\label{claim:all_left}
		Let $S$ be a minimal leftmost ($X$, $Y$, $\leq k$)$^G$-separator, and $S'$ be the separator generated by Algorithm \ref{algo:1}. Then, $ S \subseteq V_{X, S'} \cup S'$. 
	\end{claim}
	\begin{proof}
		For the sake of contradiction, assume that $\exists v \in S$ such that $v \notin V_{X, S'} \cup S'$. This means that $v \in V_{S', Y} \cup V_Z$. 
		
		If $v \in V_Z$, then $S \setminus \{v\}$ is still a leftmost ($X$, $Y$, $\leq k$)$^G$-separator which contradicts the minimality of $S$.
		
		The other possibility is that $v \in V_{S', Y}$. Let $S_{out}$ be the set of all such vertices; i.e., $S_{out} = \{v \in S \mid v \in V_{S', Y}\}$. Any $X - S_{out}$-path\footnote{any path from a vertex in $X$ to a vertex in $S_{out}$}  goes through $S'$, otherwise $S'$ would not be a separator. Let $S'_{in}$ be the set of all vertices of $S'$ that are on a path from vertex of $X$ to a vertex of $S_{out}$. Now, let $S'' = (S \setminus S_{out}) \cup S'_{in}$. Hence, $S''$ is an $(X, Y, \leq k)^{G}$-separator\footnote{TODO, why is at most $\leq k$} such that $S'' \preceq S$ with $|S''| \leq S \leq k$, which is a contradicts that $S$ is a leftmost ($X$, $Y$, $\leq k$)-separator.\hfill $\square$
	\end{proof}
	Claim~\ref{claim:all_left} allows us to ignore the subgraph $G[V_{S, Y} \cup V_Z]$ and focus only on the graph to the left of the current separator and keep moving towards left until it is impossible.
	
	Let $S_x$ be the separator found by Algorithm \ref{algo:1} while processing node $x$ from the computation tree.  If $S_x$ is a leftmost ($X$, $Y$, $\leq k$)-separator\footnote{Later, we explain how this is done}, this branch terminates and we add $S_x$ to the set of all the leftmost ($X$, $Y$, $\leq k$)$^G$-separators, namely $\mathcal{A}$ ($\mathcal{A}$ is a global variable). Otherwise, we keep pushing the separator to the left by branching 2-fold. Let us call the children of $x$ by $c_1$ and $c_2$. If $S_x$ was not a leftmost ($X$, $Y$, $\leq k$)$^{G_x}$-separator\footnote{note that by Claim~\ref{claim:all_left} $G_x = G[V_{X, S} \cup S]$}, this means that there exists at least one leftmost ($X$, $Y$, $\leq k$)$^{G_x}$-separator, namely $S'$ such that $S' \neq S$ and $S' \preceq S$. Also, as a result of Claim \ref{claim:all_left}, $S'$ is a leftmost ($X$, $Y$, $\leq k$)$^G$-separator, too.
	
	In each node $x$, we call Algorithm \ref{algo:1} to find the minimum separator $S_x$ of size $\leq k$ between $X$ and $Y_x$. Now, we push all the vertices of $S_x$ onto stack $R$. Then, we pop vertex $v$ which is on top of the stack $R$ and consider the following two scenarios (corresponding to $c_1$ and $c_2$).
	\begin{enumerate}
		\item If we want $v$ to belong to the leftmost separators. In this case, we add $v$ to set $I$\footnote{\textbf{I}nclude}, which is the set of vertices that we require to be in all the leftmost separators in $\mathcal{T}_{c_1}$.
		\item If $v$ does not belong to the leftmost separators. Here, we pop $v$ and push the left neighbors of $v$ (i.e., $N(v) \cap V_{X, S_X}$) onto $R$ (we just move to the left due to Claim \ref{claim:all_left}).
	\end{enumerate}
	Notice that the order of handling vertices is not important but we use stack because it simplifies the proof later on. 

	Every produced separator is a leftmost ($X$, $Y$, $\leq k$)$^G$-separator because the only time that one branch terminates is when it finds a leftmost separator. Now, we show that all the leftmost ($X$, $Y$, $\leq k$)$^G$-separators are generated by the algorithm given.
	
	Let $S_0$ be an arbitrary leftmost ($X$, $Y$, $\leq k$)$^G$-separator, and as before let $S$ be the separator generated by Algorithm \ref{algo:1}. At this point $R$ is filled with the vertices of $S$. Pop $v$ from $R$. 
	\begin{itemize}
		\item If $v \in S_0$, then put $v$ in $I$, and recurs.
		\item If $v \notin S_0$, push $N(v) \cap V_{X,S}$ into $R$, recurs. 
	\end{itemize}
	This determines an exact computation branch. All branches halt with a minimal leftmost separator since each time we go at least one more to the left. So, this branch terminates with minimal leftmost separator $S'_0$ as well. All the vertices of $S_0$ are pushed into $R$ at some point because otherwise this branch terminates with a separator which is not leftmost and we can push it more to the left. At the end of this branch, $I \subseteq S'_0$. On the other hand,  $I = S_0 \subseteq S'_0$, which implies that $S_0 = S'_0$ because both of them should be minimal. \hfill $\square$
\end{proof}

In \cite{marx2014fixed}, the authors mention that the number of important $(X, Y, \leq k)^G$-separators is $\leq 4^k$, and they even mention that this upper bound can be tight by a polynomial factor. Here, we give a precise upper bound for the number of leftmost $(X, Y, \leq k)^G$-separators and the number of important $(X, Y, \leq k)^G$-separators and we show that both bounds are tight.

Before that, let us review the definition of the Catalan numbers.
\begin{definition}
	\emph{Catalan numbers} is a sequence of numbers where the $n$-th Catalan number is:
	\[
	C_n = {2n \choose n} - {2n \choose n+1} = \frac{1}{n+1}{2n \choose n} \sim \frac{4^n}{n^{3/2}\sqrt{\pi}}.
	\]
	
\end{definition}
\begin{theorem}\label{thm:bound}
	Let $G=(V, E)$ be a graph and let $X, Y \subseteq V$ and let $k \in \mathbb{N}$. The number of leftmost $(X, Y, \leq k)^G$-separators is at most $C_{k-1}$ and  the number of important $(X, Y, \leq k)^G$-separators is at most $\sum\limits_{i=0}^{k-1}C_{i}$. Furthermore, both upper bounds are tight.
\end{theorem}

Before giving a proof, we mention some definitions.

\begin{definition}
	A \emph{full $k$-paranthesization} is a string over the alphabet $\{[,]\}$ consisting of $k$ ``$[$'' and	$k$ ``$]$'' having more ``$[$'' than ``$]$'' in every non-trivial prefix. This forms a language slightly different from the Dyck language (a well-known language). We call this language \emph{restricted Dyck language} and any string in this language is called a \emph{restricted Dyck word}.
\end{definition}
Notice that $[x]$ is a restricted Dyck word iff $x$ is a Dyck word\footnote{Notice the difference between the restricted Dyck Language and the Dyck language itself. In every non-trivial prefix of a Dyck word, the number of ``$[$'' is $\geq $ the number of ``$]$'', whereas for the restricted version it should be strictly greater.}. 

Full paranthesizations (the restricted Dyck language over the alphabet $\{[,]\}$) are generated by the following context-free grammar $G = (\{S, A\}, \{[,]\}, R, S)$, where the set of the rules $R$ is:
\begin{align*}
	  S &\rightarrow [A] \, | \, \epsilon\\
	  A &\rightarrow AA \, | \, [A] \, | \, \epsilon
\end{align*}

\begin{definition}
	The \emph{$k$-parentheses tree} is the binary tree $(V, E)$ with the following properties.
	\begin{itemize}
		\item $V$ is the set of prefixes of full $k'$-parenthesizations for $0 < k' \leq k$.
		\item If ``$x$'' is in $V$, then ``$x[$'' is the left child of $x$, and if ``$x]$'' is in $V$, then ``$x]$'' is the right child of $x$.
		\item If ``$x[$'' or ``$x]$'' are not in $V$, then $x$ has no left or right child
		respectively.
	\end{itemize}
\end{definition}

\begin{definition}
	The \emph{compact $\leq k$-parentheses tree} is obtained from the $\leq k$-parentheses tree by removing all nodes containing $k$ ``$[$'' and ending in
	``$]$''.
\end{definition}

Note that the nodes with $k$ ``$[$'' ending in ``$]$'' have been removed, because no branching happens at their parents, as there are no left children.

The compact $\leq k$-parentheses tree is a full binary tree, every node has $0$
or $2$ children. The number of leaves in this tree is equal to the number
of full $\leq k$-parenthesizations, which is equal to $\sum_{k'=1}^{k-1} C_{k’}
= \mathcal{O}(C_{k})$ where $C_{k}$ is the $k$-th Catalan number. An immediate
consequence is that the number of nodes in the compact $\leq k$-parentheses
tree is $2 \sum_{k’=1}^{k-1} C_{k’} -1 = \mathcal{O}(C_{k})$.

For the algorithm to find all leftmost ($X, Y, \leq k$)-separators, the worst case computation tree is the compact $\leq k$-parentheses tree with the nodes
representing full $k’$-parenthesizations for $1\leq k’ \leq k-1$ removed.

An arbitrary computation tree for the algorithm to find all leftmost
$\leq k$-separators is obtained form the worst case computation tree by the
following two operations.
\begin{enumerate}[(a)]
	\item For any node $x$, the subtree rooted at the left child $c_1(x)$ may be just removed or replaced by the subtree rooted at $c_1(c_1(x))$. This splicing operation can be repeated to jump down arbitrarily far.
	\item For any node $x$ with less than or equal to $k$ ``$[$''s, let $x_i$ be $x$ concatenated with $i$ ``$]$''s. If
	$x_{i^*}$ is a full $k’$-parenthesization for some $k’$ with $0 < k’ < k$, and the
	nodes $x_{i’}$ for $0 \leq i’< i^*$ have no left children, then the leaf $x_{i^*}$ is added. (In the added node $x_{i^*}$, a minimal separator $S$ of size $<k$ is picked, because there is no larger minimal $\leq k$-separator to the left of $S$.)
\end{enumerate}

Note that the number of leaves in an arbitrary computation tree for the	algorithm to find all ($X, Y, \leq k$)-separators is $\leq C_{k-1}$, because for every leaf ``$x]$'' added in (b), at least one leaf, namely ``$x[\cdots[]\cdots]]$'' (full parenthesization of length $k$) has been removed in (a).

For the recursive procedure, we introduce a boolean parameter leftmost. The parameter is originally true in the root and in every node that is a left child. If a node has a left child, i.e., excluding $v$, a new larger separator $S’$ with $|S’| \leq k$ has been found, then leftmost is set to false. The current truth value of leftmost is passed to the right child (include $v$ call). If $x$ contains $k’<k$ ``$[$'' and $k’-1$ ``$]$'', then a right call (including $v$) is only made if leftmost is true, i.e., there is no larger ($X, Y, \leq k$)-separator to the left of the current small (size $< k$) separator.

In every node, $S$ is the minimum leftmost separator between and including
$X$ and the current $Y$.

\begin{itemize}
    	\item In the root, $Y$ is the original $Y$. $I=\emptyset$, and $\mathit{leftmost} = true$. The node is
	represented by a sequence of ``$[$''s of length $|S|$.
	\item A node $x$ with $|S|=k’<k$, is represented by a prefix $x$ of a full parenthesization with $k’$ ``$[$'' and $|I|$ ``$]$''.
\end{itemize}

[Alternatively, one could include a root in the computation tree
corresponding to the empty $0$-parenthesization. In such a root, no $S$ is
defined. The minimum leftmost separator $S’$ is computed. A left call is
made if $|S’|\leq k$. No right call is made. In this view,
the root node has only one child. Not including such a root node means
that this node is not handled by the recursive procedure, but by a
different calling procedure.]

Now, here is the proof of Theorem \ref{thm:bound}.
\begin{proof}
	First, we prove the result corresponding to the leftmost separators. As explained above, we show excluding a vertex by a ``$[$'', which corresponds to a left branch. Analogously, ``$]$'' denotes including a vertex in $I$ (requiring the leftmost separator to have $v$ in it), which corresponds to the right branch. 
	
	Now, we see why we used a stack to handle branching on the vertices. Even though the order does not matter but we need a stack to have a nice correspondence between the algorithm behavior and the restricted Dyck words. It is well known that the number of the Dyck words of length $2k$ is $C_k$. So, the upper bound on the number of leftmost ($X$, $Y$, $\leq k$)-separators is $C_{k-1} = \frac{{2(k-1) \choose k-1}}{k}\sim\frac{4^{k-1}}{(k-1)^{3/2}\sqrt{\pi}}$.  We also close the gap and show that this bound is tight.
	
	Let $\hat{G}$ be a complete binary tree with depth at least $k+1$. Let $X$ be the set of all the leaves of $\hat{G}$, and $Y$ be the root. (see Figure \ref{fig:03-03}). The number of full subtrees with exactly $k$ leaves is $C_{k-1}$. On the other hand in $\hat{G}$, no ($X$, $Y$, $\leq k-1$)$^{\hat{G}}$-separator is a leftmost ($X$, $Y$, $\leq k$)$^{\hat{G}}$-separator since we can replace a node with its both children and move more to the left. So, the number of minimal leftmost ($X$, $Y$, $\leq k$)$^{\hat{G}}$-separators is exactly $C_{k-1}$.
	
	For the upper bound on the number of important separators, we use Lemma~\ref{lem:numleft_numimportant} and this immediately implies that  the number of important ($X$, $Y$, $\leq k$)$^{\hat{G}}$-separators $\leq \sum\limits_{i=1}^{i=k-1}C_i$.
	
	For the tightness part, assume the same $\hat{G}$ with the same $X$ and $Y$. The number of important $(X, Y, \leq k)^{\hat{G}}$-separators is $\sum\limits_{i=1}^{k-1}C_i$ since based on Lemma~\ref{lem:numleft_numimportant}, the number of important $(X, Y, \leq k)^{\hat{G}}$-separators is $\leq \sum\limits_{i=1}^{k-1}(\text{the number of leftmost $(X, Y, \leq k)^{\hat{G}}$-separators})$. On the other hand, as we mentioned earlier, in $\hat{G}$, no leftmost $(X, Y, \leq k-1)^{\hat{G}}$-separator is a leftmost $(X, Y, \leq k)^{\hat{G}}$-separator. Hence, the number of important $(X, Y, \leq k)^{\hat{G}}$-separators is exactly $ \sum\limits_{i=1}^{k-1}(\text{the number of leftmost $(X, Y, \leq k)^{\hat{G}}$-separators})$, which is equal to $\sum\limits_{i=1}^{k-1}C_i$.
	
	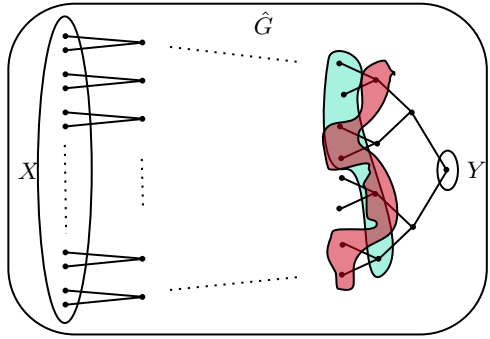
\begin{figure}
		
		\begin{minipage}[c]{0.55\textwidth}
			\tikzset{every picture/.style={line width=0.75pt}} %set default line width to 0.75pt        
			
			\begin{tikzpicture}[x=0.48pt,y=0.48pt,yscale=-1,xscale=1]
				
				%Shape: Polygon Curved [id:ds5610834250741723] 
				\draw  [fill={rgb, 255:red, 80; green, 227; blue, 194 }  ,fill opacity=0.5 ] (266,47) .. controls (269,36) and (300,43.92) .. (289.5,70.96) .. controls (279,98) and (333,190) .. (310,216) .. controls (287,242) and (302,131) .. (278,135) .. controls (254,139) and (263,58) .. (266,47) -- cycle ;
				%Shape: Circle [id:dp6314054778968607] 
				\draw  [fill={rgb, 255:red, 0; green, 0; blue, 0 }  ,fill opacity=1 ] (59,30.5) .. controls (59,29.67) and (59.67,29) .. (60.5,29) .. controls (61.33,29) and (62,29.67) .. (62,30.5) .. controls (62,31.33) and (61.33,32) .. (60.5,32) .. controls (59.67,32) and (59,31.33) .. (59,30.5) -- cycle ;
				%Shape: Circle [id:dp2215978429680936] 
				\draw  [fill={rgb, 255:red, 0; green, 0; blue, 0 }  ,fill opacity=1 ] (59,41.5) .. controls (59,40.67) and (59.67,40) .. (60.5,40) .. controls (61.33,40) and (62,40.67) .. (62,41.5) .. controls (62,42.33) and (61.33,43) .. (60.5,43) .. controls (59.67,43) and (59,42.33) .. (59,41.5) -- cycle ;
				%Shape: Circle [id:dp4345802973913748] 
				\draw  [fill={rgb, 255:red, 0; green, 0; blue, 0 }  ,fill opacity=1 ] (59,60.5) .. controls (59,59.67) and (59.67,59) .. (60.5,59) .. controls (61.33,59) and (62,59.67) .. (62,60.5) .. controls (62,61.33) and (61.33,62) .. (60.5,62) .. controls (59.67,62) and (59,61.33) .. (59,60.5) -- cycle ;
				%Shape: Circle [id:dp0620370519186364] 
				\draw  [fill={rgb, 255:red, 0; green, 0; blue, 0 }  ,fill opacity=1 ] (59,71.5) .. controls (59,70.67) and (59.67,70) .. (60.5,70) .. controls (61.33,70) and (62,70.67) .. (62,71.5) .. controls (62,72.33) and (61.33,73) .. (60.5,73) .. controls (59.67,73) and (59,72.33) .. (59,71.5) -- cycle ;
				%Shape: Circle [id:dp5431999852530733] 
				\draw  [fill={rgb, 255:red, 0; green, 0; blue, 0 }  ,fill opacity=1 ] (59,90.5) .. controls (59,89.67) and (59.67,89) .. (60.5,89) .. controls (61.33,89) and (62,89.67) .. (62,90.5) .. controls (62,91.33) and (61.33,92) .. (60.5,92) .. controls (59.67,92) and (59,91.33) .. (59,90.5) -- cycle ;
				%Shape: Circle [id:dp6309068403258316] 
				\draw  [fill={rgb, 255:red, 0; green, 0; blue, 0 }  ,fill opacity=1 ] (59,101.5) .. controls (59,100.67) and (59.67,100) .. (60.5,100) .. controls (61.33,100) and (62,100.67) .. (62,101.5) .. controls (62,102.33) and (61.33,103) .. (60.5,103) .. controls (59.67,103) and (59,102.33) .. (59,101.5) -- cycle ;
				%Straight Lines [id:da20743349418214097] 
				\draw [line width=0.75]  [dash pattern={on 0.84pt off 2.51pt}]  (60,116) -- (61,186) ;
				%Shape: Circle [id:dp6454693213149798] 
				\draw  [fill={rgb, 255:red, 0; green, 0; blue, 0 }  ,fill opacity=1 ] (59,200.5) .. controls (59,199.67) and (59.67,199) .. (60.5,199) .. controls (61.33,199) and (62,199.67) .. (62,200.5) .. controls (62,201.33) and (61.33,202) .. (60.5,202) .. controls (59.67,202) and (59,201.33) .. (59,200.5) -- cycle ;
				%Shape: Circle [id:dp03908797997540625] 
				\draw  [fill={rgb, 255:red, 0; green, 0; blue, 0 }  ,fill opacity=1 ] (59,211.5) .. controls (59,210.67) and (59.67,210) .. (60.5,210) .. controls (61.33,210) and (62,210.67) .. (62,211.5) .. controls (62,212.33) and (61.33,213) .. (60.5,213) .. controls (59.67,213) and (59,212.33) .. (59,211.5) -- cycle ;
				%Shape: Circle [id:dp5895515949530126] 
				\draw  [fill={rgb, 255:red, 0; green, 0; blue, 0 }  ,fill opacity=1 ] (59,230.5) .. controls (59,229.67) and (59.67,229) .. (60.5,229) .. controls (61.33,229) and (62,229.67) .. (62,230.5) .. controls (62,231.33) and (61.33,232) .. (60.5,232) .. controls (59.67,232) and (59,231.33) .. (59,230.5) -- cycle ;
				%Shape: Circle [id:dp629653005798553] 
				\draw  [fill={rgb, 255:red, 0; green, 0; blue, 0 }  ,fill opacity=1 ] (59,241.5) .. controls (59,240.67) and (59.67,240) .. (60.5,240) .. controls (61.33,240) and (62,240.67) .. (62,241.5) .. controls (62,242.33) and (61.33,243) .. (60.5,243) .. controls (59.67,243) and (59,242.33) .. (59,241.5) -- cycle ;
				%Shape: Ellipse [id:dp9142390677585928] 
				\draw   (39,135.5) .. controls (39,68.95) and (48.4,15) .. (60,15) .. controls (71.6,15) and (81,68.95) .. (81,135.5) .. controls (81,202.05) and (71.6,256) .. (60,256) .. controls (48.4,256) and (39,202.05) .. (39,135.5) -- cycle ;
				%Shape: Circle [id:dp1623102501838518] 
				\draw  [fill={rgb, 255:red, 0; green, 0; blue, 0 }  ,fill opacity=1 ] (119,35.5) .. controls (119,34.67) and (119.67,34) .. (120.5,34) .. controls (121.33,34) and (122,34.67) .. (122,35.5) .. controls (122,36.33) and (121.33,37) .. (120.5,37) .. controls (119.67,37) and (119,36.33) .. (119,35.5) -- cycle ;
				%Shape: Circle [id:dp8472920299558135] 
				\draw  [fill={rgb, 255:red, 0; green, 0; blue, 0 }  ,fill opacity=1 ] (119,65.5) .. controls (119,64.67) and (119.67,64) .. (120.5,64) .. controls (121.33,64) and (122,64.67) .. (122,65.5) .. controls (122,66.33) and (121.33,67) .. (120.5,67) .. controls (119.67,67) and (119,66.33) .. (119,65.5) -- cycle ;
				%Shape: Circle [id:dp18262754314937935] 
				\draw  [fill={rgb, 255:red, 0; green, 0; blue, 0 }  ,fill opacity=1 ] (119,95.5) .. controls (119,94.67) and (119.67,94) .. (120.5,94) .. controls (121.33,94) and (122,94.67) .. (122,95.5) .. controls (122,96.33) and (121.33,97) .. (120.5,97) .. controls (119.67,97) and (119,96.33) .. (119,95.5) -- cycle ;
				%Straight Lines [id:da39832165960989596] 
				\draw [line width=0.75]  [dash pattern={on 0.84pt off 2.51pt}]  (120,127) -- (121,166) ;
				%Shape: Circle [id:dp4105462482265585] 
				\draw  [fill={rgb, 255:red, 0; green, 0; blue, 0 }  ,fill opacity=1 ] (119,205.5) .. controls (119,204.67) and (119.67,204) .. (120.5,204) .. controls (121.33,204) and (122,204.67) .. (122,205.5) .. controls (122,206.33) and (121.33,207) .. (120.5,207) .. controls (119.67,207) and (119,206.33) .. (119,205.5) -- cycle ;
				%Shape: Circle [id:dp22263660291689735] 
				\draw  [fill={rgb, 255:red, 0; green, 0; blue, 0 }  ,fill opacity=1 ] (119,235.5) .. controls (119,234.67) and (119.67,234) .. (120.5,234) .. controls (121.33,234) and (122,234.67) .. (122,235.5) .. controls (122,236.33) and (121.33,237) .. (120.5,237) .. controls (119.67,237) and (119,236.33) .. (119,235.5) -- cycle ;
				%Shape: Circle [id:dp47100711743599644] 
				\draw  [fill={rgb, 255:red, 0; green, 0; blue, 0 }  ,fill opacity=1 ] (356,135.5) .. controls (356,134.67) and (356.67,134) .. (357.5,134) .. controls (358.33,134) and (359,134.67) .. (359,135.5) .. controls (359,136.33) and (358.33,137) .. (357.5,137) .. controls (356.67,137) and (356,136.33) .. (356,135.5) -- cycle ;
				%Shape: Ellipse [id:dp30590937827121456] 
				\draw   (366.5,136) .. controls (366.5,127.44) and (362.92,120.5) .. (358.5,120.5) .. controls (354.08,120.5) and (350.5,127.44) .. (350.5,136) .. controls (350.5,144.56) and (354.08,151.5) .. (358.5,151.5) .. controls (362.92,151.5) and (366.5,144.56) .. (366.5,136) -- cycle ;
				%Rounded Rect [id:dp18316255330058584] 
				\draw   (16,57) .. controls (16,28.28) and (39.28,5) .. (68,5) -- (337,5) .. controls (365.72,5) and (389,28.28) .. (389,57) -- (389,213) .. controls (389,241.72) and (365.72,265) .. (337,265) -- (68,265) .. controls (39.28,265) and (16,241.72) .. (16,213) -- cycle ;
				%Straight Lines [id:da774423002233199] 
				\draw    (60.5,30.5) -- (120.5,35.5) ;
				%Straight Lines [id:da7406752540204742] 
				\draw    (62,41.5) -- (122,35.5) ;
				%Straight Lines [id:da08404032696032693] 
				\draw    (62,60.5) -- (122,65.5) ;
				%Straight Lines [id:da12393672044605442] 
				\draw    (60.5,71.5) -- (122,65.5) ;
				%Straight Lines [id:da9948861383449887] 
				\draw    (62,90.5) -- (122,95.5) ;
				%Straight Lines [id:da7313294475280272] 
				\draw    (60.5,101.5) -- (122,95.5) ;
				%Straight Lines [id:da27200511024601837] 
				\draw    (60.5,200.5) -- (122,205.5) ;
				%Straight Lines [id:da534126802537352] 
				\draw    (62,211.5) -- (122,205.5) ;
				%Straight Lines [id:da41864846805528466] 
				\draw    (62,241.5) -- (122,235.5) ;
				%Straight Lines [id:da5545898509746032] 
				\draw    (62,230.5) -- (122,235.5) ;
				%Shape: Circle [id:dp7600340384408237] 
				\draw  [fill={rgb, 255:red, 0; green, 0; blue, 0 }  ,fill opacity=1 ] (329,90.5) .. controls (329,89.67) and (329.67,89) .. (330.5,89) .. controls (331.33,89) and (332,89.67) .. (332,90.5) .. controls (332,91.33) and (331.33,92) .. (330.5,92) .. controls (329.67,92) and (329,91.33) .. (329,90.5) -- cycle ;
				%Shape: Circle [id:dp8427749093370149] 
				\draw  [fill={rgb, 255:red, 0; green, 0; blue, 0 }  ,fill opacity=1 ] (330,180.5) .. controls (330,179.67) and (330.67,179) .. (331.5,179) .. controls (332.33,179) and (333,179.67) .. (333,180.5) .. controls (333,181.33) and (332.33,182) .. (331.5,182) .. controls (330.67,182) and (330,181.33) .. (330,180.5) -- cycle ;
				%Straight Lines [id:da9501477387248691] 
				\draw    (330.5,90.5) -- (357.5,135.5) ;
				%Straight Lines [id:da7631636071236669] 
				\draw    (331.5,180.5) -- (357.5,137) ;
				%Straight Lines [id:da7428206652056732] 
				\draw    (302.5,63.5) -- (329.5,89.25) ;
				%Straight Lines [id:da11618980824393121] 
				\draw    (303.5,115) -- (329.5,90.11) ;
				%Straight Lines [id:da753009089914429] 
				\draw    (303.5,154.5) -- (330.5,180.25) ;
				%Straight Lines [id:da3153764623403257] 
				\draw    (304.5,206) -- (330.5,181.11) ;
				%Shape: Circle [id:dp39629617174724197] 
				\draw  [fill={rgb, 255:red, 0; green, 0; blue, 0 }  ,fill opacity=1 ] (300.5,154.5) .. controls (300.5,153.67) and (301.17,153) .. (302,153) .. controls (302.83,153) and (303.5,153.67) .. (303.5,154.5) .. controls (303.5,155.33) and (302.83,156) .. (302,156) .. controls (301.17,156) and (300.5,155.33) .. (300.5,154.5) -- cycle ;
				%Shape: Circle [id:dp7905067249754476] 
				\draw  [fill={rgb, 255:red, 0; green, 0; blue, 0 }  ,fill opacity=1 ] (301,65) .. controls (301,64.17) and (301.67,63.5) .. (302.5,63.5) .. controls (303.33,63.5) and (304,64.17) .. (304,65) .. controls (304,65.83) and (303.33,66.5) .. (302.5,66.5) .. controls (301.67,66.5) and (301,65.83) .. (301,65) -- cycle ;
				%Shape: Circle [id:dp2404186657721683] 
				\draw  [fill={rgb, 255:red, 0; green, 0; blue, 0 }  ,fill opacity=1 ] (302,115) .. controls (302,114.17) and (302.67,113.5) .. (303.5,113.5) .. controls (304.33,113.5) and (305,114.17) .. (305,115) .. controls (305,115.83) and (304.33,116.5) .. (303.5,116.5) .. controls (302.67,116.5) and (302,115.83) .. (302,115) -- cycle ;
				%Shape: Circle [id:dp3124737113326641] 
				\draw  [fill={rgb, 255:red, 0; green, 0; blue, 0 }  ,fill opacity=1 ] (303,205.5) .. controls (303,204.67) and (303.67,204) .. (304.5,204) .. controls (305.33,204) and (306,204.67) .. (306,205.5) .. controls (306,206.33) and (305.33,207) .. (304.5,207) .. controls (303.67,207) and (303,206.33) .. (303,205.5) -- cycle ;
				%Straight Lines [id:da9983864902573907] 
				\draw    (274.5,141) -- (301.5,153.5) ;
				%Straight Lines [id:da10543743222828428] 
				\draw    (275.5,166) -- (301.5,153.92) ;
				%Straight Lines [id:da8971739951683302] 
				\draw    (276.5,193) -- (303.5,205.5) ;
				%Straight Lines [id:da8861274411155278] 
				\draw    (277.5,218) -- (303.5,205.92) ;
				%Shape: Circle [id:dp7922485943699775] 
				\draw  [fill={rgb, 255:red, 0; green, 0; blue, 0 }  ,fill opacity=1 ] (274.5,218) .. controls (274.5,217.17) and (275.17,216.5) .. (276,216.5) .. controls (276.83,216.5) and (277.5,217.17) .. (277.5,218) .. controls (277.5,218.83) and (276.83,219.5) .. (276,219.5) .. controls (275.17,219.5) and (274.5,218.83) .. (274.5,218) -- cycle ;
				%Straight Lines [id:da3646209216745604] 
				\draw    (275.5,52) -- (302.5,64.5) ;
				%Straight Lines [id:da4135500489346613] 
				\draw    (276.5,77) -- (302.5,64.92) ;
				%Straight Lines [id:da16221941870490686] 
				\draw    (274.5,102) -- (301.5,114.5) ;
				%Straight Lines [id:da06200698741086064] 
				\draw    (275.5,127) -- (301.5,114.92) ;
				%Shape: Circle [id:dp68483080843601] 
				\draw  [fill={rgb, 255:red, 0; green, 0; blue, 0 }  ,fill opacity=1 ] (275,194.5) .. controls (275,193.67) and (275.67,193) .. (276.5,193) .. controls (277.33,193) and (278,193.67) .. (278,194.5) .. controls (278,195.33) and (277.33,196) .. (276.5,196) .. controls (275.67,196) and (275,195.33) .. (275,194.5) -- cycle ;
				%Shape: Circle [id:dp3427933020908327] 
				\draw  [fill={rgb, 255:red, 0; green, 0; blue, 0 }  ,fill opacity=1 ] (272.5,166) .. controls (272.5,165.17) and (273.17,164.5) .. (274,164.5) .. controls (274.83,164.5) and (275.5,165.17) .. (275.5,166) .. controls (275.5,166.83) and (274.83,167.5) .. (274,167.5) .. controls (273.17,167.5) and (272.5,166.83) .. (272.5,166) -- cycle ;
				%Shape: Circle [id:dp8747630905884738] 
				\draw  [fill={rgb, 255:red, 0; green, 0; blue, 0 }  ,fill opacity=1 ] (274.5,142) .. controls (274.5,141.17) and (275.17,140.5) .. (276,140.5) .. controls (276.83,140.5) and (277.5,141.17) .. (277.5,142) .. controls (277.5,142.83) and (276.83,143.5) .. (276,143.5) .. controls (275.17,143.5) and (274.5,142.83) .. (274.5,142) -- cycle ;
				%Shape: Circle [id:dp3685456280588515] 
				\draw  [fill={rgb, 255:red, 0; green, 0; blue, 0 }  ,fill opacity=1 ] (274,126.5) .. controls (274,125.67) and (274.67,125) .. (275.5,125) .. controls (276.33,125) and (277,125.67) .. (277,126.5) .. controls (277,127.33) and (276.33,128) .. (275.5,128) .. controls (274.67,128) and (274,127.33) .. (274,126.5) -- cycle ;
				%Shape: Circle [id:dp42155190100971685] 
				\draw  [fill={rgb, 255:red, 0; green, 0; blue, 0 }  ,fill opacity=1 ] (273,101.5) .. controls (273,100.67) and (273.67,100) .. (274.5,100) .. controls (275.33,100) and (276,100.67) .. (276,101.5) .. controls (276,102.33) and (275.33,103) .. (274.5,103) .. controls (273.67,103) and (273,102.33) .. (273,101.5) -- cycle ;
				%Shape: Circle [id:dp07049816143258036] 
				\draw  [fill={rgb, 255:red, 0; green, 0; blue, 0 }  ,fill opacity=1 ] (276,76.5) .. controls (276,75.67) and (276.67,75) .. (277.5,75) .. controls (278.33,75) and (279,75.67) .. (279,76.5) .. controls (279,77.33) and (278.33,78) .. (277.5,78) .. controls (276.67,78) and (276,77.33) .. (276,76.5) -- cycle ;
				%Shape: Circle [id:dp4419141012912917] 
				\draw  [fill={rgb, 255:red, 0; green, 0; blue, 0 }  ,fill opacity=1 ] (272.5,52) .. controls (272.5,51.17) and (273.17,50.5) .. (274,50.5) .. controls (274.83,50.5) and (275.5,51.17) .. (275.5,52) .. controls (275.5,52.83) and (274.83,53.5) .. (274,53.5) .. controls (273.17,53.5) and (272.5,52.83) .. (272.5,52) -- cycle ;
				%Curve Lines [id:da9005070058219791] 
				\draw [color={rgb, 255:red, 0; green, 0; blue, 0 }  ][fill={rgb, 255:red, 208; green, 2; blue, 27 }  ,fill opacity=0.5 ][line width=0.75] [line join = round][line cap = round]   (317,61) .. controls (317,58.11) and (315.36,58.7) .. (314,57) .. controls (313.9,56.88) and (310.46,50.15) .. (310,50) .. controls (306.07,48.69) and (297.47,53.53) .. (295,56) .. controls (285.84,65.16) and (297.45,84.55) .. (288,94) .. controls (280.64,101.36) and (266.88,94.6) .. (262,106) .. controls (260.51,109.47) and (261.41,115.04) .. (262,118) .. controls (262.3,119.51) and (264.46,134.49) .. (266,135) .. controls (274.42,137.81) and (292.81,128.21) .. (297,131) .. controls (298.11,131.74) and (296.68,133.71) .. (297,135) .. controls (297.26,136.02) and (298.33,137) .. (298,138) .. controls (297.61,139.16) and (291.9,146.9) .. (293,148) .. controls (297.05,152.05) and (294.03,167.74) .. (296,173) .. controls (296.94,175.52) and (302.15,174.45) .. (303,177) .. controls (306.27,186.82) and (273.71,179.29) .. (269,184) .. controls (267.74,185.26) and (267.87,187.82) .. (267,190) .. controls (261.65,203.37) and (266.55,207.1) .. (271,216) .. controls (272.22,218.44) and (271.13,228.53) .. (273,229) .. controls (293.45,234.11) and (280.83,218.33) .. (287,206) .. controls (295.51,188.98) and (316.87,194.33) .. (319,173) .. controls (320.82,154.78) and (311.23,134.23) .. (301,124) .. controls (297.81,120.81) and (289.65,117.53) .. (289,113) .. controls (287.61,103.29) and (292.03,102.97) .. (298,97) .. controls (304.41,90.59) and (309.76,81.95) .. (312,73) .. controls (312.49,71.03) and (312.09,68.81) .. (313,67) .. controls (314.33,64.33) and (312,62) .. (317,59) ;
				%Straight Lines [id:da0494281095707918] 
				\draw  [dash pattern={on 0.84pt off 2.51pt}]  (142,39) -- (244,51) ;
				%Straight Lines [id:da606172956064529] 
				\draw  [dash pattern={on 0.84pt off 2.51pt}]  (142,230) -- (241,220) ;
				
				% Text Node
				\draw (205,9) node [anchor=north west][inner sep=0.75pt]   [align=left] {$\displaystyle \hat{G}$};
				% Text Node
				\draw (21,128) node [anchor=north west][inner sep=0.75pt]   [align=left] {$\displaystyle X$};
				% Text Node
				\draw (371,127) node [anchor=north west][inner sep=0.75pt]   [align=left] {$\displaystyle Y$};

			\end{tikzpicture}
		\end{minipage}\hfill
		\begin{minipage}[c]{0.43\textwidth}
			\caption{
				Here, $k = 6$. We illustrate two out of $42$ minimal leftmost ($X$, $Y$, $\leq 6$)$^{\hat{G}}$-separators. Also, there are $1 + 2 + 5 + 14 + 42 = 64$ important ($X$, $Y$, $\leq 6$)$^{\hat{G}}$-separators.
			} \label{fig:03-03}
		\end{minipage}
	\end{figure}
	\hfill $\square$
	\end{proof}
	
	Notice that $\sum\limits_{i=1}^{k-1}C_i = \Theta(\frac{4^{k-1}}{(k-1)^{3/2}})$. 
	
	\begin{remark} Let $\mathcal{A}_{X, Y, \leq}^G$ and $\mathcal{B}_{X, Y, \leq}^G$ be the set of all leftmost ($X$, $Y$, $\leq k$)$^G$-separators and the set of   
	Based on Lemma~\ref{lem:numleft_numimportant}, the set of all important ($X$, $Y$, $\leq k$)$^G$-separators that are not leftmost ($X$, $Y$, $\leq k$)$^G$-separator is the union of the set of all leftmost ($X$, $Y$, $\leq i$)$^G$-separator, for $i=1, \cdots, k-1$.
	\end{remark}
	\begin{remark}
	Let $\hat{G}$ be a complete binary tree with depth at least $k+1$, $\hat{X}$ be the set of all leaves of $\hat{G}$, and $\hat{Y}$ be the root. For a fixed $k \in \mathbb{N}$,
	\begin{enumerate}
	    \item $\hat{G}$ has the highest number of leftmost ($X$, $Y$, $\leq k$)$^G$-separators among all graphs like $G$ and for all $X, Y \subseteq V(G)$ and it happens by setting $X = \hat{X}$ and $Y = \hat{Y}$.
	    \item $\hat{G}$ has the highest number of important ($X$, $Y$, $\leq k$)$^G$-separators among all graphs like $G$ and for all $X, Y \subseteq V(G)$ and it happens by setting $X = \hat{X}$ and $Y = \hat{Y}$.
	    \item $\hat{G}$ has the highest number of important but not leftmost ($X$, $Y$, $\leq k$)$^G$-separators among all graphs like $G$ and for all $X, Y \subseteq V(G)$ and it happens by setting $X = \hat{X}$ and $Y = \hat{Y}$.
	\end{enumerate}
	\end{remark}
	\begin{theorem}
		Let $G = (V, E)$ be a graph, $X, Y \subseteq V$, and $k \in \mathbb{N}$. There is an algorithm which finds all minimal leftmost ($X$, $Y$, $\leq k$)$^G$-separators in time $\mathcal{O}(\frac{4^k}{\sqrt{k}}n)$.
	\end{theorem}
	\begin{proof}
		Our algorithm searches all the possibilities and enumerates all the possible leftmost separators. We showed that the number of leaves, which is the number of leftmost separators is $\mathcal{O}(C_{k-1})$. So, we have $\mathcal{O}(C_{k-1})$ nodes in our computation tree and work done in every node is $\mathcal{O}(kn)$ (the running time of a simple flow algorithm). Hence, the total running time is $\mathcal{O}(\frac{2^{2k}}{\sqrt{k}}n)$.\hfill $\square$
	\end{proof}
	\section{Application to Treewidth Approximation}
	As mentioned in the introduction, Reed's  $5$-approximation algorithm \cite{reed1992finding} for treewidth  runs in time $\mathcal{O}(2^{24k}(k+1)!\,n \log n)$. This was improved to a $5$-approximation algorithm running in $\mathcal{O}(k^2 2 ^{8.766k} n \log n)$ \cite{belbasi2020improvement}. Here, we further improve the treewidth approximation algorithm. In order to do so, we briefly describe the algorithms given in \cite{reed1992finding}, and \cite{belbasi2020improvement}. Hence, we review some notations and for the others we give references. 

For a graph $G = (V, E)$ and a subset $W$ of the vertices, $G[W]$ is the subgraph induced by $W$. For the sake of simplicity throughout this paper, let $G - W$ be $G[V\setminus W]$ and $G - v$ be $G - \{v\}$ for any $W \subseteq V(G)$ and any $v \in V(G)$.\\
Also, in a weighted graph, a non-negative integer weight $w(v)$ is defined for each vertex $v$. For a subset $W$ of the vertices, the weight $w(W)$ is simply the sum of the weights of all vertices in $W$. Furthermore, the total weight or the weight of $G$ is the weight of $V$.
\begin{definition}
    A \emph{tree decomposition} of a graph $G = (V, E)$ is a tree $\mathcal{T} = (V_{\mathcal{T}}, E_{\mathcal{T}})$ such that each node $x$ in $V_{\mathcal{T}}$ is associated with a set $B_x$ (called the bag of $x$) of vertices in $G$, and $\mathcal{T}$ has the following properties:
\begin{itemize}
    \item $\bigcup\limits_{x \in V_{\mathcal{T}}} B_x = V(G)$
    
    \item $\forall \{u, v\} \in E, \, \exists x\in V_{\mathcal{T}}: u,v \in B_x$
    
    \item $\forall x, y \in V_{\mathcal{T}}, \, \forall z\in V_{\mathcal{T}} \text{ on the path connecting $x, y \in V_{\mathcal{T}}$: $B_x \cap B_y \subseteq B_z.$}$
\end{itemize}

Even though historically the \emph{width of a tree decomposition} is  defined to be the size of its largest bag minus one, in this paper we follow Reed's definition \cite{reed1992finding} and define the {\emph width of a tree decomposition} to be the size of its largest bag\footnote{For the sake of simplicity in the computation}. 

The \emph{treewidth} of a graph $G$ is the minimum width over all tree decompositions of $G$ called $tw(G)$. In the following, we use the letter $k$ for the treewidth.
\end{definition}
\begin{definition}
    A \emph{nice tree decomposition} is a tree decomposition which is rooted and its every node has at most two children. Any node $x$ in a nice tree decomposition $\mathcal{T}$ is of one of the following types (let $c$ be the only child of $x$ or let $c_1$ and $c_2$ be the two children of $x$):
\begin{itemize}
\item \emph{Leaf} node, a leaf of $\mathcal{T}$ without any children.
\item \emph{Forget} node (forgetting vertex $v$), where $v \in B_c$ and $B_x = B_c \setminus \{v\}$,
\item \emph{Introduce vertex} node (introducing vertex $v$), where $v \notin B_c$ and $B_x = B_c \cup \{v\}$ ,
\item \emph{Join} node, where $x$ has two children with the same bag as $x$, i.e. $B_x = B_{c_1} = B_{c_2}$.
\end{itemize}
\end{definition}
It has been shown that any given tree decomposition can be converted to a nice tree decomposition with the same width in polynomial time\cite{kneis2009bound}.

\begin{definition}
	A centroid of a weighted tree $T$ is a node $x$ such that none of the trees in the forest $T-x$ has more than half the total weight.
\end{definition}
For nice tree decompositions\footnote{The property that we use is that the bags of the adjacent nodes in a nice tree decomposition differ by at most one vertex. In fact, any tree decoposition with adjacent bags differing in at most one vertex works just fine.}, we choose a stronger version of centroid for this paper.
\begin{definition}
	A strong centroid of a nice tree decomposition $\mathcal{T}$ of a graph $G = (V, E)$ with respect to $W \subseteq V$ is a node $x$ of $\mathcal{T}$ such that none of the connected components of $G - B_x$ contains more than $\frac{1}{2}|W\setminus B_x|$ vertices of $W$.
\end{definition}
The following lemma shows there existence of a strong centroid for any given $W\subseteq V$.
\begin{lemma}\label{lem:cent}
	For every nice tree decomposition $(\mathcal{T}, \{B_x: x\in V_{\mathcal{T}}\})$ of a graph $G = (V, E)$ and every subset $W \subseteq V$, there exist a strong centroid with respect to $W$.
\end{lemma}
The proof of Lemma~\ref{lem:cent} can be found in \cite{belbasi2020improvement}.
\begin{definition} Let $G = (V,E)$ be a graph and $W \subseteq V$. A \emph{balanced $W$-separator} is a set $S \subseteq V$ such that every connected component of $G-S$ has at most $\frac{1}{2}|W|$ vertices.
\end{definition}

\begin{lemma}\label{Lem:ResBalSep}{\rm \cite[Lemma 11.16]{FlumGrohe}}
	Let $G = (V,E)$ be a graph of treewidth at most $k-1$ and $W \subseteq V$. Then there exists a balanced $W$-separator of $G$ of size at most $k$.
\end{lemma}

\begin{definition}
	Let $G = (V,E)$ be a graph and $W \subseteq V$. A \emph{weakly balanced separation} of $W$ is a triple $(X, S, Y )$, where $X, Y \subseteq W$, $S \subseteq V$ are pairwise disjoint sets such that:
	\begin{itemize}
		\item $W = X \cup (S \cap W) \cup Y$.
		\item $S$ separates $X$ from $Y$.
		\item $0 < |X|, |Y| \leq \frac{2}{3}|W|$.
	\end{itemize}
\end{definition}

\begin{lemma}\label{lem:exs_weak}{\rm \cite[Lemma 11.19]{FlumGrohe}}
	For $k \geq 3$, let $G = (V, E)$ be a graph of treewidth at most $k-1$ and $W \subseteq V$ with $|W| \geq 2k+1$. Then there exists a weakly balanced separation of $W$ of size at most $k$.
\end{lemma}

\begin{theorem}{\rm \cite[Corollary 11.22]{FlumGrohe}}
	For a graph of treewidth at most $k-1$ with a given set $W \subseteq V$ of size $|W| = 3k-2$, a weakly balanced separation of $W$ can be found in time $O(2^{3k} k^2 n)$.
\end{theorem}

\subsection{Summary of Reed's algorithm}

Reed finds a weakly balanced separator for the set of ``representatives'' as follows (it has been described in \cite{reed1992finding} and the details have been filled in \cite{belbasi2020improvement}):
\begin{itemize}
	\item Do a DFS on $G$ and find the deepest vertex $v\in V(G)$ whose subtree has at least $\frac{n}{24k}$ vertices. Let us call $v$ a representative. Also, define the weight of $v$ to be the size of this subtree, denoted as $w(v)$.
	\item Let $W$ be the set of all representatives.
	\item Form a balanced separator $S$ of size $\leq k$, partitioning $G - S$ into $X$ and $Y$ as follows:
	\begin{itemize}
		\item Decide if any representative is going to go into $S$. Branch 2-fold and consider both scenarios.
		\begin{itemize}
			\item Scenario \rom{1}: At least a representative like $v$ is going into $S$. Place $v$ into $S$, decrease $k$ by one, do a DFS from scratch and form a new set of representatives and recurs.
			\item Scenario \rom{2}: Branch on every representative like $v$ going into $X$ or $Y$. The crux of the idea is that if a representative goes into one side, most probably its corresponding subtree will also go into the same side. It is because if $v$ goes into let us say $X$ and one of the proper descendants of one of its children in its corresponding subtree like $u$ goes into $Y$, then the path from $v$ to $u$ should go through $S$, which means it requires a vertex from its subtree to be in $S$. However, the size of $S$ is bounded by $k$. So, not more than $k$ such incidents can happen. This means at most $k\cdot  \frac{n}{24k}$ vertices might go to the opposite side of their corresponding representative (notice that no children of $v$ has $\geq \frac{n}{24k}$ vertices in its subtree). This means the error is $\leq \frac{1}{24}n$. This property allows Reed to work with the representatives (there are  $\leq 24k$ of them), rather than all the vertices. Branching on every vertex results in an exponential-time algorithm in term of $n$ but it does not happen when we work with the representatives.
			
			Notice that we separated the branching on $X$ and $Y$ from the branching on $S$. It is because if a representative goes into $S$, we have no control over its corresponding  subtree. 
			
			Because of the error mentioned above, $\frac{1}{3}n - \frac{1}{24}n\leq w(X),w(Y)\leq\frac{2}{3}n + \frac{1}{24}n$.
			
			These bounds are in terms of weight while we want a balanced separator in terms of the exact size. So, in order to go from a separator in terms of weight to a separator in terms of the volume (actual size), we might have to pay up to the error, once more. So, the separator that we find for $W$, separates the entire graph into $L$, and $R$ such that $\frac{1}{3}n - \frac{1}{24}n - \frac{1}{24}n\leq |L|,|R|\leq\frac{2}{3}n + \frac{1}{24}n + \frac{1}{24}n$, which means $\frac{1}{4}n \leq |L|, |R| \leq \frac{3}{4}n$. 
		\end{itemize}
	\end{itemize}
\end{itemize}
Once Reed obtains $X$ and $Y$, uses the flow algorithm to find the leftmost minimum size ($X$, $Y$, $\leq k$)-separator, which is unique.

Now, $G-S$ is separated into two sides $L$ and $R$, such that $\frac{1}{4}n \leq |L|, |R|\leq \frac{3}{4}n$. We recurs on subproblems with the inputs $G[X\cup S]$ and $G[Y\cup S]$. Then, the algorithm finds a tree decomposition for each subproblem and finally merges them together to obtain a tree decomposition for $G$.  

Based on Lemma \ref{lem:cent} there exists a strong centroid with respect to $W$. We do not necessarily know what it is but we sure know that it exists. Reed's algorithm checks all the partitions and one of them is in fact the centroid. So, this is why the algorithm definitely finds a balanced separator if it exists.
\subsection{Further Previous Improvement}
Recently, the authors of this paper, improved Reed's algorithm so that it runs in time $\mathcal{O}(k^2 2^{8.766k}n \log n)$. Here, we briefly describe how this improvement has been achieved.

As mentioned earlier, Reed picks $24$ but any constant $C_{\epsilon} > 6$ gives us two subproblems with sizes $\epsilon\, n$ and $(1 - \epsilon) n$ (or better), for any small and nonzero $\epsilon$. In Reed's algorithm, $\epsilon = \frac{1}{4}$, and $C_{\epsilon} = 24$. We can fix $\epsilon$ and $C_{\epsilon}$, later.

There are two major improvements in \cite{belbasi2020improvement}. 
\begin{itemize}
	\item First improvement is that we do not separate scenarios \rom{1} and \rom{2} as in Reed's algorithm. We consider one representative at a time and once we are working on representative $v$, we branch on $v$ going into the separator or not. If it goes into the separator, we put $v$ there but this time we do not do the DFS from scratch to form a new group of representatives (which is costly). We undo the DFS for the subtree rooted at $v$ (the subtree that $v$ represents) and continue the DFS.
	\item Second, we do not construct a bipartition at the very beginning after removing $S$. Reed starts with looking for a weakly balanced separator ($\frac{1}{3}n$ to $\frac{2}{3} n$), which is known to exist. However, he works with the weights and as argued before, he makes a bipartition based on the weights but the volumes can be $\frac{1}{4} n$ to $\frac{3}{4} n$ (or better). In \cite{belbasi2020improvement}, our argument starts with a restricted balanced separator by volume. We allow more than two partitions (multi partition by volume for the analysis). Each part has volume of at most $\frac{1}{2}$. This means each part has weight up to $3/4 - \epsilon/2$. Only now, we form a two-partition (after applying the weights) by weight. Each part has weight $\leq \max \{3/4 - \epsilon/2, 2/3\}$ (for $\epsilon \leq 1/6$). Then, we find a separator for this weight partition. The larger part has volume between $1/4 + \epsilon$ and $1 - \epsilon$. So, in worst-case scenario, we end up with an $\epsilon$ to $1 - \epsilon$ partition, which is still good. Hence, we set $C_{\epsilon} = \frac{1}{((1-\epsilon) - 1/2)/2} = \frac{4}{1-2\epsilon} \leq 4 + 12\epsilon$, for $\epsilon \leq 1/6$.
\end{itemize}

Our time analysis in \cite{belbasi2020improvement} shows that the running time complexity to split based on $W$ is as given below, where $t$ is the number of subtrees (or representatives), and $k$ is the upper bound on treewidth.   
\begin{align}
	T_{W}(t, k) \leq T_{W}(t, k-1) + 2T_{W}(t-1, k) + Qkn + \mathcal{O}(1),
\end{align}
where $Q$ is the constant which shows up in the time complexity of the flow algorithm. We have shown that finally,
\begin{align}
	T_{W}(t, k) \leq Qk^2ne^{k}(C_{\epsilon} + 1)^{k}(2^{(C_{\epsilon} + 1)k}(C_{\epsilon}+1)+2).
\end{align}
The following theorem can be found in \cite{belbasi2020improvement}.
\begin{theorem}
	Let $C_0 = \log_2 e + \log_2 5 + 5 < 8.765$. For every $C > C_0$, a 5-approximation algorithm of the treewidth can be computed in time $\mathcal{O}(2^{Ck}n \log n)$.
\end{theorem}

	\subsection{Our Improvement}
	In this subsection, our goal is to use the algorithm for finding the leftmost separators to further improve the coefficient of $k$ in the exponent of the tree decomposition algorithm to make it more applicable.
	
	\begin{itemize}
		\item 
		For the analysis, we consider a centroid by volume, namely $C$. It has size $k' \leq k$.
		\item
		Each connected component of $G-C$ has volume at most $\frac{1}{2} (n-k')$. 
		\item
		These connected components can be grouped into 3 parts, each with volume at most $\frac{1}{2} (n-k')$.
		(Just place the components by decreasing volume into the part with currently smallest volume.)
		%% Each part has weight at most 1 - e, but this is not needed.
		\item
		Let the proper volume be the part of the volume that has its corresponding weight in the same part.
		In other words, the proper volume is the number of vertices whose representative is in the same part.
		\item
		Let $t$ be the threshold for the size of the small trees.
		At most $k' (t-2) = k'((\frac{1}{2} - \epsilon)n / k - 2) \leq (\frac{1}{2} - \epsilon) n - 2k'$ vertices 
		%% Before, I used k' instead of k here.
		can be in a different part than their representative.
		Therefore, the total proper volume is at least $n - k' - (\frac{1}{2} - \epsilon) n + 2k' \geq (\frac{1}{2} + \epsilon) n + k'$
		\item
		Of the proper volume, at least $(\frac{1}{2} + \epsilon) n + k' - \frac{1}{2} (n-k') > \epsilon n$  %%+ \frac{3k'}{2}$ 
		is not in the part with largest proper volume.
		\item
		Therefore, there are at least 2 parts with proper volume at least $\frac{\epsilon n}{2}$ %% + \frac{k}{4}$.
		\item
		Of these 2 parts, we put the part with larger weight on the left side, the other one on the right side.
		\item
		We also put the third part on the left side.
		\item
		The left part has weight at least half the total weight, which is $\frac{1}{2} (n-k')$.
		\item
		The right part has weight at most $\frac{1}{2} (n-k')$ and (proper) volume at least $\frac{\epsilon n}{2}$ %% + \frac{k}{4}$. 
	\end{itemize}
	
	The algorithm tries all possible 2-partitions of the representatives. This includes the left-right partition that we are currently investigating.
	While searching for a leftmost separator, the centroid is a competitor. Thus the algorithm finds a separator that is equal to the centroid or is located strictly to the left of it. 
	From now on, left (call it $X$) and right (namely $Y$) are defined by the leftmost ($X$, $Y$, $\leq\! k$)$^G$-separator found by Algorithm~\ref{algo:2}. This separator has size $k''$ with $k' \leq k'' \leq k$. It produces the same weight partition as the centroid, but part of the volume might shift to the right.
	%% I have assumed here, that the vertices in the centroid have no weight. The total weight is $n-k'$.
	
	\begin{itemize}
		\item
		Thus the left part has still weight at least $\frac{1}{2} (n-k')$ and therefore volume at least $\frac{1}{2} (n-k') - (\frac{1}{2} - \epsilon) n + 2k
		\geq \epsilon n + \frac{3}{2} k$.
		\item
		The right part has still volume at least $\frac{\epsilon n}{2}$ %% + \frac{k}{4}$.
		\item
		The recursive calls are done with the subgraphs induced by the union of the vertices of a connected component with the vertices of the separator. Their number of vertices is upper bounded by $n$ minus the volume of the smaller side. It is less than $(1 - \frac{\epsilon}{2})n$.
		\item 
		After $\frac{2 \ln 2}{\epsilon} (\log n - \log b) = O(\log n)$ rounds for $b \geq k$, the largest volume of a recursive call is at most $b$.
	\end{itemize}
	
	In the worst case, the algorithm alternates between a split by volume and $\log k$ splits of $W$ steps. Let the time spent between two splits by volume be at most $f(k)n$. Note that $f(k) \leq g(k) + 3^{3k} k \log k = O(g(k))$, where $g(k)$ is the time of one split by volume step.
	Then we get the following recurrence for an upper bound on the running time of the whole algorithm.
	\[
	T(n) = 
	\begin{cases}
		O(k)	&	\text{if $n \leq 3k$} \\
		\max_{p, n_1,\dots, n_p} \left\{ \sum_{i \in [p]} T(n_i) \right\} + f(k)n	&	\text{otherwise,} \\
	\end{cases}
	\]
	where the maximum is taken over $p \geq 2$ and $n_1, \dots, n_p \in [n-1]$ such that $\sum_{i \in [p]} (n_i - k) = n - k$.
	Note that $\sum_{i \in [p]} (n_i - k'') = n - k''$ reflects that every recursive call includes a connected component of $G-S$ together with the separator $S$ of size $k''$. We can round up $k''$ to $k$, because $T(n)$ is an increasing function. Because, the sum of the $n_i$'s is more than $n$, it is beneficial to consider the following modified function $T'(n') = T(n+k)$. Then we get the simpler recursion
	\[
	T'(n') = 
	\begin{cases}
		O(k)	&	\text{if $n' \leq 2k$} \\
		\max_{p, n'_1,\dots, n'_p} \left\{ \sum_{i \in [p]} T'(n'_i) \right\} + f(k)(n'+k)	&	\text{otherwise,} \\
	\end{cases}
	\]
	where the maximum is taken over $p \geq 2$ and $n'_1, \dots, n'_p \in [n'-1]$ such that $\sum_{i \in [p]} (n'_i) = n'$.
	
	Now we prove
	\[T'(n') \leq \frac{c}{\epsilon}  f(k) n' \log n' \]
	by induction, where $c$ is minimal such that $c \geq 3$ and the base case ($n' \leq 2k$) is satisfied. 
	Assume that the $i$th component of size $n_i$ is on the side of the separator with smaller volume if and only if $1 \leq i \leq p'$. 
	Let $n_S = \sum_{i \in [p']} n'_i $, and let $n_L = n' - n_S$. Furthermore, let 
	\[h_S = \sum_{i \in [p']} n'_i \log n'_i  \leq  \sum_{i \in [p']} n'_i \log \frac{n'}{2} = n_S (\log n' - 1), \]
	and
	\[h_L  = \sum_{i \in \{p'+1, \dots, p\}} n'_i \log n'_i  \leq  \sum_{i \in  \{p'+1, \dots, p\}} n'_i \log n' = n_L \log n'. \]
	
	Recall that $\epsilon n' /2 \leq n_S \leq n'/2$.
	By the inductive hypothesis, for $n'>2k$ we have 
	\[ T'(n') \leq \frac{c}{\epsilon} f(k) (h_S +h_L) + f(k)(n'+k) \]  
	
	This implies
	\[ T'(n') \leq \frac{c}{\epsilon} f(k) (n' \log n' -  n_S ) + f(k)(n'+k) \]
	Thus $T'(n') \leq  \frac{c}{\epsilon} f(k) n' \log n'$ if $\frac{c}{\epsilon} n_S \geq n' +k$. As $n_S \geq \epsilon n' / 2$ and $n'>2k$, this is the case when $c \geq 3$.
	
	Each split by volume can be done by finding at most $C_{k-1} = \Theta(4^k / k^{3/2})$ separators in time $\mathcal{O}(C_{k-1} kn) = \mathcal{O}(4^k / \sqrt{k})$  for each placement of at most $n/t = k/(1/2 - \epsilon)$ representatives to the left or right side and the placement of at most $k$ vertices into the centroid. These are at most  $(2 +8 \epsilon) k$ representatives for $\epsilon \leq 1/4$. Choosing $\epsilon = \Theta(1/k)$ this results in a running time of $f(k) n = \mathcal{O}({{(3+8 \epsilon)k} \choose k} 2^{4k} k^{-1/2} n) = \mathcal{O}(\frac{3^{3k}}{2^{2k}k} 2^{4k}  n$ for one split by volume in a graph of size $n$. Together with the solution of the previous recurrence, we obtain.
	
	\begin{theorem}
		If a graph has treewidth at most $k$, then a tree decomposition of width at most $4k-1$ can be found in time
		$\mathcal{O}(2^{6.755k} n \log n)$.
	\end{theorem}
	
	\begin{proof}
		Here, we analyze the running time of the treewidth algorithm that we just described.
		
		First, we check the running time of split by $W$, denoted by $T_W(k,n)$. Since $|W| \leq 3k$, and $|S| \leq k$, then we have at most $k$ placements into the centroid and $2k$ outside the centroid. Thus,
		\begin{eqnarray*}
			T_W(k, n) & \approx &		\binom{3k}{k} 2^{2k} \frac{4^k}{\sqrt{k}}n\\
			& \approx & 	\frac{3^3k}{2^2k} 2^{2k} \frac{4^k}{\sqrt{k}}n\\
			& = & 		2^{(3 \log_2 3 + 2)k} \frac{1}{\sqrt{k}}n\\
			& = &		\frac{2^{2k}  3^{3k}}{\sqrt{k}} n\\
			& < &		\frac{2^{6.755 k}}{\sqrt{k}} n
		\end{eqnarray*}
		Therefore, $T_W(k,n) = \mathcal{O}(2^{6.755 k}k^{-1/2} n)$.
		
		Now, we do a split by $V$ every $\log_2 k$ step. So, the total running time of our algorithm is:
		\[\mathcal{O}(\frac{2^{6.755 k}}{\sqrt{k}}\,k\, \log k\, n \,\log n) = \mathcal{O}(2^{6.755 k}(\log k) \sqrt{k} n (\log n)).\]
		Notice that we rounded th exponent up, hence the polynomial part is dominated by this roundup and the total running time is $\mathcal{O}(2^{6.755k}n \log n)$.\hfill $\square$
	\end{proof}
	
	\bibliographystyle{amsplain}
	\bibliography{refs}

\providecommand{\bysame}{\leavevmode\hbox to3em{\hrulefill}\thinspace}
\providecommand{\MR}{\relax\ifhmode\unskip\space\fi MR }
\providecommand{\MRhref}[2]{%
  \href{http://www.ams.org/mathscinet-getitem?mr=#1}{#2}
}
\providecommand{\href}[2]{#2}
\begin{thebibliography}{10}

\bibitem{amir2010approximation}
Eyal Amir, \emph{Approximation algorithms for treewidth}, Algorithmica
  \textbf{56} (2010), no.~4, 448--479.

\bibitem{arnborg1987complexity}
Stefan Arnborg, Derek~G Corneil, and Andrzej Proskurowski, \emph{Complexity of
  finding embeddings in a k-tree}, SIAM Journal on Algebraic Discrete Methods
  \textbf{8} (1987), no.~2, 277--284.

\bibitem{belbasi2020improvement}
Mahdi Belbasi and Martin F{\"u}rer, \emph{An improvement of {R}eed's treewidth
  approximation}, 2020.

\bibitem{bodlaender2016c}
Hans~L Bodlaender, P{\aa}l~Gr{\o}n{\aa}s Drange, Markus~S Dregi, Fedor~V Fomin,
  Daniel Lokshtanov, and Micha{\l} Pilipczuk, \emph{A c\^{}kn 5-approximation
  algorithm for treewidth}, SIAM Journal on Computing \textbf{45} (2016),
  no.~2, 317--378.

\bibitem{chitnis2013fixed}
Rajesh Chitnis, MohammadTaghi Hajiaghayi, and D{\'a}niel Marx,
  \emph{Fixed-parameter tractability of directed multiway cut parameterized by
  the size of the cutset}, SIAM Journal on Computing \textbf{42} (2013), no.~4,
  1674--1696.

\bibitem{courcelle1990monadic}
Bruno Courcelle, \emph{The monadic second-order logic of graphs. {I}.
  recognizable sets of finite graphs}, Information and computation \textbf{85}
  (1990), no.~1, 12--75.

\bibitem{FlumGrohe}
J.~Flum and M.~Grohe, \emph{Parameterized complexity theory ({T}exts in
  theoretical computer science. an {EATCS} series)}, Springer-Verlag, Berlin,
  Heidelberg, 2006.

\bibitem{kneis2009bound}
Joachim Kneis, Daniel M{\"o}lle, Stefan Richter, and Peter Rossmanith, \emph{A
  bound on the pathwidth of sparse graphs with applications to exact
  algorithms}, SIAM Journal on Discrete Mathematics \textbf{23} (2009), no.~1,
  407--427.

\bibitem{korhonen2021single}
Tuukka {Korhonen}, \emph{{A single-exponential time 2-approximation algorithm
  for treewidth}}, arXiv e-prints (2021), arXiv:2104.07463.

\bibitem{lagergren1996efficient}
Jens Lagergren, \emph{Efficient parallel algorithms for graphs of bounded
  tree-width}, Journal of Algorithms \textbf{20} (1996), no.~1, 20--44.

\bibitem{lokshtanov2013clustering}
Daniel Lokshtanov and D{\'a}niel Marx, \emph{Clustering with local
  restrictions}, Information and Computation \textbf{222} (2013), 278--292.

\bibitem{marx2006parameterized}
D{\'a}niel Marx, \emph{Parameterized graph separation problems}, Theoretical
  Computer Science \textbf{351} (2006), no.~3, 394--406.

\bibitem{marx2014fixed}
D{\'a}niel Marx and Igor Razgon, \emph{Fixed-parameter tractability of multicut
  parameterized by the size of the cutset}, SIAM Journal on Computing
  \textbf{43} (2014), no.~2, 355--388.

\bibitem{menger1927allgemeinen}
Karl Menger, \emph{Zur allgemeinen {K}urventheorie}, Fundamenta Mathematicae
  \textbf{10} (1927), no.~1, 96--115.

\bibitem{reed1992finding}
Bruce~A Reed, \emph{Finding approximate separators and computing tree width
  quickly}, Proceedings of the twenty-fourth annual ACM symposium on Theory of
  computing (STOC), 1992, pp.~221--228.

\bibitem{robertson1995graph}
Neil Robertson and P.~D. Seymour, \emph{Graph minors. {XIII}. {T}he disjoint
  paths problem}, Journal of combinatorial theory, Series B \textbf{63} (1995),
  no.~1, 65--110.

\end{thebibliography}
	\newpage

\end{document}